\documentclass[twocolumn,10pt]{IEEEtran}
\usepackage{caption}
\usepackage{subcaption}
\usepackage{cite}
\usepackage{amsmath,amssymb,amsfonts,amsthm}
\usepackage{graphicx}
\usepackage{algorithm,algorithmic}
\usepackage{color}
\def\BibTeX{{\rm B\kern-.05em{\sc i\kern-.025em b}\kern-.08em
    T\kern-.1667em\lower.7ex\hbox{E}\kern-.125emX}}
\newtheorem{definition}{Definition}[section]
\newtheorem{problem}{Problem}[section]
\newtheorem{assumption}{Assumption}[section]
\newtheorem{theorem}{Theorem}[section]
\newtheorem{lemma}{Lemma}[section]

\newtheorem{remark}{Remark}[section]
\DeclareMathOperator{\image}{image}
\DeclareMathOperator{\rank}{rank}
\DeclareMathOperator{\tr}{tr}
\DeclareMathOperator{\sgn}{sgn}
\DeclareMathOperator{\diag}{diag}

\newcommand{\nullspace}{\operatorname{null}}

\usepackage{tikz}
\usetikzlibrary{arrows,shapes,chains,calc}
\usetikzlibrary{arrows.meta}
\usepackage{tikz-3dplot}
\usepackage{tikz}

\newenvironment{apxproof}{%
  \par
  \normalfont
  \trivlist
  \item[]\hspace*{2pt} 
}{%
  \hfill $\blacksquare$ 
  \endtrivlist
}

\makeatletter

\newcommand{\Rmnum}[1]{\expandafter\@slowromancap\romannumeral #1@}
\makeatother

\newif\ifshowmarkup
\showmarkupfalse

\ifshowmarkup
    \newcommand{\HT}[1]{{\color{blue}#1}}
\else
    \newcommand{\HT}[1]{#1}
\fi

\begin{document}
\title{Distributed Non-Uniform Scaling Control of Multi-Agent Formation via Matrix-Valued Constraints}
\author{Tao He and Gangshan Jing
\thanks{This work was supported in part by the National Key Research and Development Program of China under Grant 2025YFA1018800 and in part by the National Natural Science Foundation of China under Grant 62573068 and Grant 62533002. {\it (Corresponding author: Gangshan Jing)}}
\thanks{Tao He and Gangshan Jing are with Chongqing University, Chongqing, 400044, PRC.  (e-mail:20231301010@stu.cqu.edu.cn; jinggangshan@cqu.edu.cn). }}

\maketitle

\begin{abstract}
Distributed formation maneuver control refers to the problem of maneuvering a group of agents to change their formation shape by adjusting the motions of partial agents, where the controller of each agent only requires local information measured from its neighbors. Although this problem has been extensively investigated, existing approaches are mostly limited to uniform scaling transformations. This article proposes a new type of local matrix-valued constraints, via which non-uniform scaling control of position formation can be achieved by tuning the positions of only two agents (i.e., leaders). Here, the non-uniform scaling transformation refers to global scaling the position formation with different ratios along different orthogonal coordinate directions. Moreover, by defining scaling and translation of attitudes, we propose a distributed control scheme for scaling and translation maneuver control of joint position-attitude formations. It is proven that the proposed controller achieves global convergence, provided that the sensing graph among agents is a 2-rooted bidirectional graph. Compared with the affine formation maneuver control approach, the proposed approach leverages a sparser sensing graph, requires fewer leaders, and additionally enables scaling transformations of the attitude formation. A simulation example demonstrates our theoretical results.



\end{abstract}


\section{Introduction}
Formation maneuver control enables a group of agents to operate as a cohesive unit, with maneuverability defined as the degree to which the formation's positional (centroid, scale, and other geometric parameters) and attitudinal (orientation) characteristics can be continuously adjusted while maintaining coordinated motion. This capability is essential for applications such as search and rescue \cite{Zhou2022, Quan2023}, cooperative transport \cite{Liu2024}, cooperative localization \cite{Chen2025}, and collaborative manipulation \cite{Culbertson2021}. However, dynamic and complex environments, such as obstacle-dense or high-interference scenarios, pose significant challenges to formation maneuverability.


\begin{figure}[t]
	\centering	
	\includegraphics[width=3in]{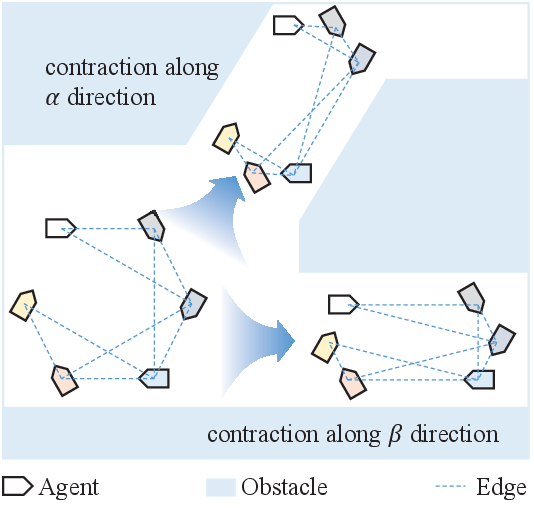}   	
	\caption{Non-uniform scaling transformation along arbitrary direction of the formation under sensing constraints in an obstacle-cluttered environment.}     
	\label{fig:fignustOA}      
\end{figure}



The maneuverability of a multi-agent position formation is fundamentally constrained by the types of local inter-agent constraints that characterize the overall formation geometry. Prior studies have demonstrated translational maneuvers via displacement-based consensus methods \cite{DeMarina2021, Romero2024}, and rotational maneuvers through inter-agent distance constraints \cite{Asimow1978, HeXiaodong2024, Vu2024}. \HT{Compared to rigid transformations only, geometrically scalable formations\footnote{\HT{This geometric scalability is distinct from the
scalability of the control algorithm with respect to the number of
agents.}}—achievable via bearing-based \cite{Trinh2020, Erskine2024}, angle-based \cite{Jing2019, Chen2023}, distance-ratio-based \cite{Cao2020}, complex-Laplacian-based \cite{Lin2014, Fang2024}, and clique-based \cite{HeGen2025} control strategies—enable uniform scaling (changing the overall formation size while preserving shape), thus significantly improve formation maneuverability.}





However, existing uniform scaling control methods are often inefficient in anisotropic settings. \HT{For example, in elongated corridors \cite{Quan2023}, uniform scaling may require unnecessary reduction along unconstrained directions, while non-uniform scaling\footnote{\HT{Here, ``non-uniform scaling'' means a single, global anisotropic scaling transformation applied identically to all agents, not multiple independent uniform scalings applied to different subsets, see Definition~\ref{de:position_transformation}.}} enables selective compression along the constrained axis, better accommodating spatial or hardware constraints.} In dynamic environments \cite{Alonso2017}, non-uniform scaling further improves responsiveness by reducing superfluous transformations, which are often time-consuming. Although affine formation control \cite{Arranz2014, Lin2016, Zhao2018, Aranda2025} theoretically enables non-uniform scaling, existing methods face challenges due to complex sensing graph structures and the high computational cost of centralized optimization over constraint matrices. As a result, non-uniform scaling transformations in formation control remain relatively underexplored.

On the other hand, attitude formation control introduces additional complexity to maneuverability. Most existing approaches either seek full heading consensus, aligning all agents to a common orientation for simplified coordination \cite{KwangKyo2014}, or aim to maintain fixed relative attitudes to preserve structured formation patterns with constant orientation differences \cite{Dimarogonas2009, Song2017, HeXiaodong2024}. These approaches rarely account for the coupled nature of position and attitude in practical scenarios, and overlook scaling transformations of attitude that could significantly enhance the formation's agility and adaptability.



To address the above-mentioned challenges, we investigate distributed strategies for non-uniform scaling of formations, as demonstrated in Fig. \ref{fig:fignustOA}. The main contributions are as follows.


\begin{enumerate}
    \item To ensure that all maneuver parameters are effective and the follower state are uniquely determined by the leader state, we introduce the concept of maximum maneuverability, and establish the necessary and sufficient graphical conditions under which a formation achieves maximum maneuverability within the leader–follower framework; see Section \ref{Maneuverability}.

    \item \HT{We design a local linear constraint and construct a matrix-valued Laplacian to characterize the target formation. An efficient method for computing the corresponding stabilizing matrix is developed (see Lemma~\ref{lem:PathG} and Theorem~\ref{the_dsm}). In contrast to existing approaches \cite{Lin2014, Zhao2018, ZHANG2025}, where the stabilizing matrix is computed for the entire formation, our method enables decentralized computation over individual dual-entry paths (formally defined in Definition \ref{def:dual_entry_path}).}

    \item \HT{We propose a distributed non-uniform scaling maneuver control law for the joint position-attitude formation. Under a 2-rooted graph structure, we guarantee global convergence of the closed-loop system (see Theorems~\ref{th convergence} and~\ref{the_dd}). In contrast to affine formation control \cite{Zhao2018, Zhao2024} (which achieves non-uniform scaling but requires a $3$-rooted graph in $\mathbb{R}^2$) and to geometrically scalable formations \cite{Lin2014, Jing2019, Cao2020, Trinh2020, Chen2023, Erskine2024, Fang2024, HeGen2025} (which are limited to uniform scaling only), our approach enables non-uniform scaling using only two leaders while requiring each follower to have merely two neighbors. Crucially, we also introduce scaling of attitudes, a capability previously unexplored in joint position–attitude formation control \cite{HeXiaodong2024, Meng2025}.}

        
\end{enumerate}

The paper is structured as follows. Section \ref{Preliminaries} introduces the preliminaries. Section \ref{ProblemStatement} formulates the formation maneuver control problem. Section \ref{Maneuverability} analyzes maximum maneuverability and its conditions. Section \ref{Control} presents the controller. Section \ref{Simulation} provides numerical validation. Section \ref{sec:discussion} discusses practical aspects, and Section \ref{Conclusion} concludes the paper.

\section{Preliminaries}\label{Preliminaries}
\subsection{Notations}
Throughout this paper, $\mathbb{R}$ denotes the set of real numbers, $\mathbb{R}^d$ the $d$-dimensional Euclidean space, $\dim(\cdot)$ the dimension of a linear space, and $|\cdot|$ the cardinality of a set or the element-wise absolute value for a scalar, vector, or matrix. Let $\nullspace(\cdot)$, $\image(\cdot)$, $\tr(\cdot)$, $\det(\cdot)$ and $\rank(\cdot)$ denote the null space, image space, trace, determinant and rank of a matrix, respectively. The identity matrix is $I_n \in \mathbb{R}^{n \times n}$, $1_n \in \mathbb{R}^n$ the all-ones vector, $0$ a zero tensor (scalar/vector/matrix) with context-appropriate dimensions, and $\otimes$ the Kronecker product.

For $x \in \mathbb{R}^d$, $\diag(x)$ denotes its diagonal matrix. For matrices $A_i$, $i=1,\dots,m$, $\diag(A_1,\dots,A_m)$ or $\diag(A_i)$ denotes the block-diagonal matrix; if all $A_i=A$, we write $\diag(A)$. The Euclidean norm is $\|\cdot\|$, while $\wedge$, $\Rightarrow$, and $\Longleftrightarrow$ denote logical conjunction, implication, and equivalence, respectively. $R(\theta) = \begin{bmatrix} \cos \theta & -\sin \theta \\ \sin \theta & \cos \theta \end{bmatrix}$ is a rotation matrix.

\subsection{Graph Theory}\label{subsec:graph_theory}
Consider a graph $G = (V, E)$ representing a multi-agent system, where the vertex set $V = \{1, \dots, n\} $ denotes the set of agents and the edge set $ E \subseteq \{(i, k) : i, k \in V \text{ and } i \neq k\} $ captures sensing relationships. Each directed edge $(i, k) \in E$ indicates that agent $k$ can measure information from agent $i$ and there is a directed edge from vertex $k$ to vertex $i$. We refer to $G$ as a sensing graph since its edges explicitly encode the directional sensing relationships between agents. The neighbor set of agent $k$ is $N_k= \{ i \in V : (i, k) \in E \}$.

\HT{A bidirectional path from agent $i_1$ to agent $i_k$ in a directed graph $G$ is a sequence of distinct agents $i_1, i_2, \dots, i_k$ such that $(i_l, i_{l+1}) \in E$ and $(i_{l+1}, i_l) \in E$ for all $l = 1, \dots, k-1$. Agents $i_1$ and $i_k$ are called the end agents, while any intermediate agents are termed inner agents. The graph $G$ is called bidirectional if, for every edge $(i,j) \in E$, its reverse edge $(j,i)$ also belongs to $E$.}


\HT{A matrix $M=[M_{ki}]\in\mathbb{R}^{nd\times nd}$ is called a matrix-valued Laplacian if $\sum_{i=1}^nM_{ki}=0$ for $k\in V$, where each block $M_{ki}=0$ if $(i,k)\notin E$. Unlike the scalar‑weighted graph Laplacian, here $M_{ki}$ is a $d\times d$ matrix whose entries can be positive, negative, or zero. Note that graph $G$ is not undirected, which implies that, in general, $M_{ki} \neq M_{ik}^{\top}$ for $i\neq k$.}





\begin{definition}[\cite{Lin2014}]
    For a bidirectional graph $G$, an agent $i$ is said to be 2-reachable from a non-singleton set $U$ of agents if there exists a bidirectional path from an agent in $U$ to agent $i$ after removing any agent except agent $i$.
\end{definition}
\begin{definition}[\cite{Lin2014}] \label{def:2-rooted}
    A bidirectional graph $G$ is said to be 2-rooted if there exists a set of two agents (called roots), from which every other agent is 2-reachable.
\end{definition}
\HT{
\begin{definition}[Dual-Entry Path] \label{def:dual_entry_path}
    A dual-entry path (DEP) is a graph $G_\mathcal{P} = (V_\mathcal{P}, E_\mathcal{P})$ comprising two distinct entry agents $i, j$ and an ordered sequence of $\ell \geq 1$ inner agents $1, \dots, \ell$ forming a bidirectional path, such that:

    \begin{itemize}
        \item If $\ell = 1$, then $V_\mathcal{P} = \{i, j, 1\}$ and $E_\mathcal{P} = \{(i, 1), (j, 1)\}$;
        \item If $\ell \geq 2$, then $V_\mathcal{P} = \{i, j\} \cup \{1, \dots, \ell\}$ and $E_\mathcal{P} = \{(i, 1), (j, \ell)\} \cup \{(k, k+1), (k+1, k) : 1 \leq k < \ell\}$.
    \end{itemize}
\end{definition}}
\begin{figure}[t!hbp]
	\centering	
	\includegraphics[width=3in]{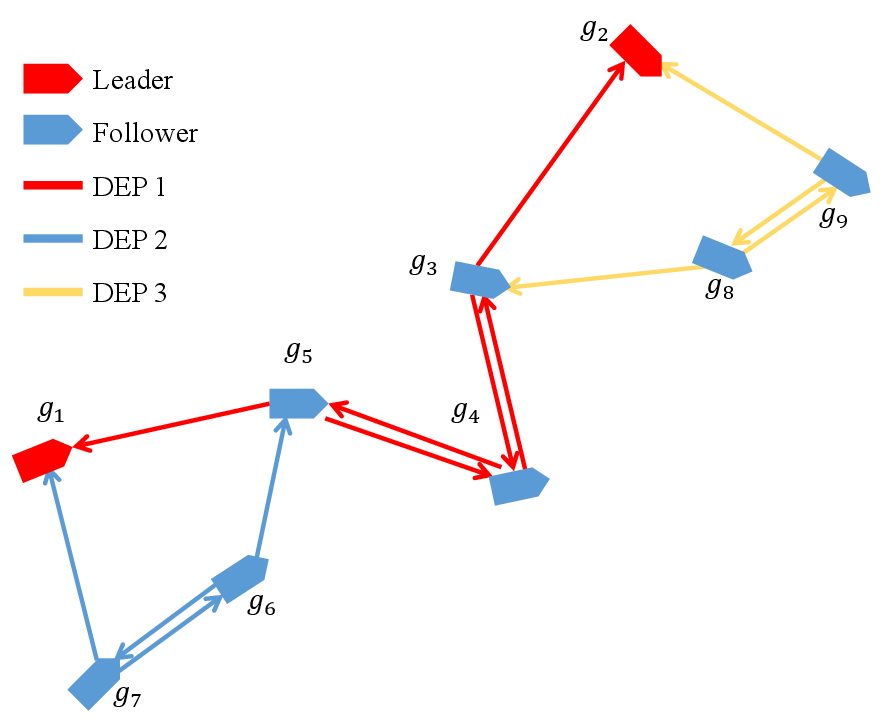}   
    \caption{\HT{DEP-induced graph formed by attaching DEPs~1--3 sequentially, illustrating hierarchical dependencies and the two-neighbor-per-follower structure. In DEP~3, agent~9 senses~2 and~8 ($9  \rightarrow 2$, $9  \rightarrow 8$, i.e., $(2,9)$, $(8,9) \in E$), agent~8 senses~3 and~9 ($8  \rightarrow 3$, $8  \rightarrow 9$, i.e., $(3,8)$, $(9,8) \in E$), but not vice versa ($3 \nrightarrow 8$, $2 \nrightarrow 9$). Hence, DEP~3 depends unidirectionally on DEP~1 via entry agents~2 and~3.}}  
        
	\label{fig_mdep}      
\end{figure}
\begin{definition}[DEP-Induced Graph] \label{def:multi_dual_entry_graph}
	Let $\mathcal{L}_0 = (V_0, E_0)$ be the graph with agent set $V_0 = \{1,2\}$ and edge set $E_0 = \emptyset$. For $ h = 1, \dots, \kappa $, define $ \mathcal{L}_h = (V_h, E_h) $ by attaching a DEP $ G_{\mathcal{P}_h} = (V_{\mathcal{P}_h}, E_{\mathcal{P}_h}) $ to $ \mathcal{L}_{h-1} = (V_{h-1}, E_{h-1}) $ via distinct entry agents $ i_h, j_h \in V_{h-1} $, where $V_{\mathcal{P}_h} \cap V_{h-1} = \{i_h, j_h\}$, the inner agents of $G_{\mathcal{P}_h}$ are labeled as $\{|V_{h-1}| + 1, \dots, |V_{h-1}| + \ell_h\}$, and $V_h = V_{h-1} \cup V_{\mathcal{P}_h}$, $E_h = E_{h-1} \cup E_{\mathcal{P}_h}$.
\end{definition}

\HT{Unlike the 2‑rooted graph in complex‑Laplacian‑based formation control \cite{Lin2014}, the proposed DEP‑induced graph 
is generated by sequentially attaching dual‑entry paths (Definition \ref{def:dual_entry_path}) to two entry agents from the existing graph. Since entry agents do not sense inner agents, mutual dependencies between adjacent DEPs—present in bidirectional paths—are eliminated; moreover, the construction ensures each follower senses exactly two neighbors.

An example of DEP-induced graph is given in Fig. \ref{fig_mdep}. The following lemma formally establishes a connection between 2-rooted graphs and the DEP-induced graph.}

\begin{lemma} \label{lem:mp2rg_equivalence}
	A bidirectional graph $G$ is 2-rooted if and only if it contains a DEP-induced graph as its spanning subgraph.
\end{lemma}

\begin{proof}
    See Appendix \ref{sec:proofLemmp2rg_equivalence}.
\end{proof}




\section{Problem Formulation}\label{ProblemStatement}
This section introduces the foundational concepts and rigorous formulation of the formation maneuver control problem.

\subsection{Joint Position-Attitude Formation}\label{subsec:jpaf}
Consider a group of $n$ agents in $\mathbb{R}^2$. The dynamics of the $k$-th agent is given by
\begin{equation}\label{equ_dy}
\dot{g}_k=\begin{bmatrix} \dot{p}_k \\
\dot{\phi}_k \end{bmatrix}= \begin{bmatrix} u_k \\
\omega_k \end{bmatrix} ,
\end{equation}
where $g_k=[p^{\top}_k, \phi_k]^{\top} \in \mathbb{R}^3$, $p_k=[p_k^x, p_k^y]^{\top} \in \mathbb{R}^2$, $\phi_k \in \mathbb{R}$, $u_k \in \mathbb{R}^2$, and $\omega_k \in \mathbb{R}$ denote the state, position, yaw angle, linear velocity and yaw rate, respectively, of agent $k$.

\HT{A formation\footnotemark[3] $(G,q)$ consists of a sensing graph $G$ and a stacked vector $q$ denoting the configuration of all agents. By defining $p=[p^{\top}_1, \cdots, p^{\top}_k, \cdots , p^{\top}_n]^{\top}$, $\phi=[\phi_1, \cdots, \phi_k, \cdots , \phi_n]^{\top}$, $g= [g^{\top}_1, \cdots, g^{\top}_k, \cdots , g^{\top}_n]^{\top}$, depending on the states that constitute $q$, we distinguish three types:
\begin{itemize}
	\item If $q=p$, $(G,p)$ is a position formation.
	\item If $q=\phi$, $(G,\phi)$ is an attitude formation.
	\item If $q = g$, $(G,g)$ is a joint position-attitude formation.
\end{itemize}


This paper studies a generalized joint position-attitude formation. Unlike position-only formation \cite{jing2018,Fang2024}, which treats agents as point masses to achieve desired spatial configurations, and attitude-only formation \cite{Zou2012,Wei2018}, which synchronizes orientations for specific directional relationships, the proposed approach enables finer control over the formation's global geometry and internal structure. Furthermore, it extends existing joint formation methods \cite{Wu2014,HeXiaodong2024,Meng2025} by supporting non‑uniform scaling of both positions and attitudes, as detailed in the following sections.

\footnotetext[3]{\HT{In references of formation shape stabilization, e.g., \cite{Anderson2009,Fathian2021}, a formation usually refers to the geometric shape formed by the group of agents, which can be invariant even when agents have different positions. In contrast, the formation in this paper is defined as a framework, which has a one-to-one correspondence with a configuration. Such a definition allows us to distinguish formations with the same shape but different locations or attitudes in formation maneuver control.}}
}

\begin{figure}[htb]
    \centering	
    \includegraphics[width=3.4in]{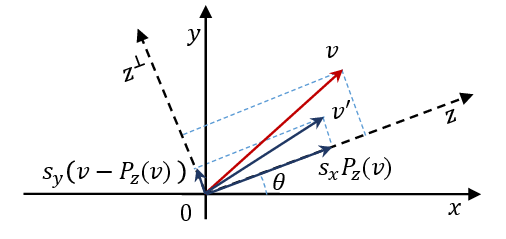}   	
    \caption{Non-uniform scaling transformation in $z=[\cos\theta, \sin\theta]^{\top}$ direction.}     
    \label{fig:fignust}      
\end{figure}
    
\subsection{Non-Uniform Scaling of a Position Formation}
As illustrated in Fig.~\ref{fig:fignust}, we first define a non-uniform scaling transformation for a vector $v \in \mathbb{R}^2$ along an arbitrary direction $z = [\cos\theta, \sin\theta]^\top$, where $\theta$ is called the scaling direction. The transformation applies scaling factors $s_x$ and $s_y$ along the axes aligned with $z$ and $z^{\perp}=R(\frac{\pi}{2})z$, respectively. The transformed vector is then given by:
\begin{equation}
\begin{aligned}
v' &= s_x P_z(v) + s_y (v - P_z(v)) \\
&= \left((s_x - s_y) zz^{\top} + s_y I_2\right) v \\
&= \left(R(\theta)
\begin{bmatrix}
s_x - s_y & 0 \\
0 & 0
\end{bmatrix}
R^{\top}(\theta) + s_y I_2\right) v \\
&= R(\theta)
\begin{bmatrix}
s_x & 0 \\
0 & s_y
\end{bmatrix}
R^{\top}(\theta) v.
\end{aligned}
\nonumber
\end{equation}
Here, $P_z(v) = zz^{\top} v$ denotes the projection of $v$ onto direction $z$. Note that when $s_x=s_y=s$, the transformation reduces to $v' =sv$, which corresponds to a uniform scaling case.

We now extend this concept to the position formation.
\begin{definition}[Non-Uniform Scaling of a Position Formation] \label{de:position_transformation}    
    Given a nominal position formation $(G,\tilde{p})$ in $\mathbb{R}^2$ with configuration $\tilde{p} = [\tilde{p}_1^{\top}, \cdots, \tilde{p}_k^{\top}, \cdots , \tilde{p}_n^{\top}]^{\top} \in \mathbb{R}^{2n}$, its non-uniform scaling transformation associated with scaling direction $\theta$ is defined as:
    \begin{equation} 
    p' = \left(I_n \otimes \left(R(\theta) \diag(s_p) R^{\top}(\theta)\right)\right) \tilde{p},
    \label{eq_nusp}
    \end{equation}
    where $s_p=[s_x, s_y]^{\top} \in \mathbb{R}^2$ is the scaling factor vector.
\end{definition}

This transformation enables continuous modulation of formation shapes along arbitrary directions, providing a foundation for the anisotropic scaling formation maneuver control strategy proposed in this paper. Compared to uniform scaling methods \cite{Lin2014,Zhao2016,jing2018} and fixed scaling approaches \cite{DeMarina2021,Vu2024,Romero2024}, the proposed non-uniform scaling offers superior flexibility in controlling multi-agent formations.

\subsection{Scaling and Translation of an Attitude Formation}\label{subsec:scaling_attitude}
\HT{
Existing attitude formation methods focus on consensus alignment \cite{KwangKyo2014} or fixed relative attitudes \cite{Dimarogonas2009, Song2017}. Although effective for basic coordination, their limited adaptability restricts their capacity to meet dynamic operational requirements. To enhance flexibility, we introduce scaling and translation transformations for attitude formation, enabling continuous modulation of attitude formation geometry.

\begin{definition}[Scaling and Translation of an Attitude Formation] \label{de attitude transformation}
    Given a nominal attitude formation $(G,\tilde{\phi})$ with configuration $\tilde{\phi} = [\tilde{\phi}_1, \cdots, \tilde{\phi}_k, \cdots , \tilde{\phi}_n]^{\top} \in [-\pi, \pi)^n$, its scaling and translation transformation is defined as:
    \begin{equation}
    \phi' = s_\phi \tilde{\phi} + \tau_\phi 1_n,
    \label{eq_sta}
    \end{equation}
    where $s_{\phi} \in \mathbb{R}$ and $\tau_{\phi} \in \mathbb{R}$ are the scaling and translation factors, respectively, and are chosen to ensure $\phi' \in [-\pi, \pi)^n $.
\end{definition}

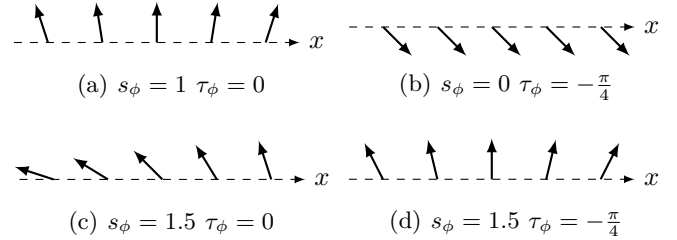
\begin{figure}[htbp]
	\centering
	\foreach \i/\title/\scale/\trans in {
		1/(a) $s_\phi=1 \  \tau_\phi = 0$ /1/0,
		2/(b) $s_\phi=0 \  \tau_\phi = -\frac{\pi}{4}$/0/-45,
		3/(c) $s_\phi=1.5 \  \tau_\phi = 0$/1.5/0,
		4/(d) $s_\phi=1.5 \  \tau_\phi = -\frac{\pi}{4}$/1.5/-45
	}{
		\begin{subfigure}[b]{0.23\textwidth}
			\centering
			\begin{tikzpicture}[scale=0.9, >=latex]
			\def\n{5}
			\pgfmathsetmacro{\sep}{0.8}
			\pgfmathsetmacro{\startx}{1.2}
			\foreach \j in {0,...,4} {
				\pgfmathsetmacro{\phi}{(\j*9+72}
				\pgfmathsetmacro{\theta}{\scale*\phi+\trans}
				\pgfmathsetmacro{\x}{\startx - \j*\sep}
				\pgfmathsetmacro{\dx}{cos(\theta)*0.6}
				\pgfmathsetmacro{\dy}{sin(\theta)*0.6}
				\draw[->, thick] (\x,0) -- ++(\dx,\dy);
			}
            \draw[->, dashed] (-2.5,0) -- (1.7,0) node[right] {$x$};
			\end{tikzpicture}
			\caption*{\title}
		\end{subfigure}
            \ifnum\i<3
              \vspace{0.5cm}
            \fi
	}
	\caption{\HT{Attitude formation transformation. Each arrow indicates an agent's yaw angle $\tilde{\phi}_k \in [-\pi, \pi)$, measured from the positive $x$-axis (counterclockwise positive, clockwise negative). (a) Original configuration $\tilde{\phi} = [\frac{3 \pi}{5}, \frac{11 \pi}{20}, \frac{\pi}{2}, \frac{9 \pi}{20}, \frac{2 \pi}{5}]$. (b) Translation only. (c) Scaling only. (d) Scaling + translation.}}
	\label{fig:tikz_attitude_transform}
\end{figure}

Definition~\ref{de attitude transformation} is demonstrated through the transformations shown in Fig.~\ref{fig:tikz_attitude_transform}(b)--(d), while subfigure (a) provides the original configuration for comparison. 

To better understand the geometric effects of the scaling factor $s_\phi$ and the translation factor $\tau_\phi$, let $\bar{\phi} = \frac{1}{n}\sum_{i=1}^n \tilde{\phi}_i$ be the center of $(G,\tilde{\phi})$. Equation~\eqref{eq_sta} can be rewritten as
\begin{equation}
    \phi' = s_\phi (\tilde{\phi} - \bar{\phi} 1_n) + \bar{\phi} 1_n + \tau_\phi' 1_n,
    \nonumber
\end{equation}
with $\tau_\phi' = \tau_\phi - \bar{\phi}(1 - s_\phi)$. This form reveals a four‑step process: (1) zero‑centering by subtracting $\bar{\phi} 1_n$, (2) scaling the deviations by $s_\phi$, (3) restoring the original center by adding $\bar{\phi} 1_n$, and (4) applying a pure translation $\tau_\phi' 1_n$.

Hence, $s_\phi$ scales only the zero‑centered deviations; $\tau_\phi = \bar{\phi}(1 - s_\phi) + \tau_\phi'$ is the sum of a compensation term $\bar{\phi}(1 - s_\phi)$, which counteracts the center shift induced by scaling when $\bar{\phi} \neq 0$ and $s_\phi \neq 1$, and the intended pure translation $\tau_\phi'$. Moreover, for any $i \neq j$,
\begin{equation}
    \phi'_i - \phi'_j = s_\phi (\tilde{\phi}_i-\tilde{\phi}_j).
    \nonumber
\end{equation}
Thus, $s_\phi$ directly scales the inter-agent yaw angle differences: its magnitude $|s_\phi|$ determines whether the differences are amplified ($|s_\phi| > 1$) or compressed ($0 \le |s_\phi| < 1$), while its sign determines whether the algebraic sign of each difference is preserved or reversed. 

Given $\tilde{\phi}$, let $\Delta = \max_i \tilde{\phi}_i - \min_i \tilde{\phi}_i$. To ensure $\phi' \in [-\pi, \pi)^n$, the scaling factor must satisfy $|s_\phi| < \frac{2\pi}{\Delta}$. For a given $s_\phi$ in this range, the translation $\tau_\phi$ then lies in $\left[ \max_i (-\pi - s_\phi \tilde{\phi}_i),\ \min_i (\pi - s_\phi \tilde{\phi}_i) \right)$.}

\subsection{Problem Statement}
In this article, we aim to achieve combined transformations including translation, and non-uniform scaling of the nominal joint position-attitude formation by tuning only the states of partial agents. As shown in Fig.~\ref{fig:fignustOA}, when avoiding obstacles, a formation that can perform a non-uniform scaling transformation in an arbitrary direction is more environmentally friendly and efficient compared to those that can only perform uniform scaling transformation in the literature \cite{Buckley2021}\cite{Gao2023}.

We adopt a leader–follower strategy for formation maneuver control. Consider a formation comprising $m$ leaders and $n-m$ followers, with the leader set denoted as $V_l=\{1, \cdots, m\}$ and the follower set as $V_f=\{m+1, \cdots, n\}$. The leader state and follower state are defined as $g_l=[g_1^{\top}, \cdots, g_m^{\top}]^{\top} \in \mathbb{R}^{3m}$ and $g_f=[g_{m+1}^{\top}, \cdots, g_n^{\top}]^{\top} \in \mathbb{R}^{3(n-m)}$, respectively. 

\subsubsection{Target Formation} \label{sss:target_formation}
We focus on the nominal joint position-attitude formation subject to non-uniform scaling along a specified direction. To explicitly represent such a setting, we extend the formation representation from the pair $(G, g)$ to a triple $(G, \tilde{g}, \theta)$, where $G$ remains the underlying sensing graph, while $(\tilde{g}, \theta)$ jointly describes an arbitrarily chosen nominal configuration for the team of agents. Specifically, the nominal state $\tilde{g} = [\tilde{g}_l^{\top}, \tilde{g}_f^{\top}]^{\top}$, where $\tilde{g}_l=[\tilde{g}_1^{\top}, \cdots, \tilde{g}_m^{\top}]^{\top} \in \mathbb{R}^{3m}$ represents the nominal state corresponding to the leaders, and $\tilde{g}_f=[\tilde{g}_{m+1}^{\top}, \cdots, \tilde{g}_n^{\top}]^{\top} \in \mathbb{R}^{3(n-m)}$ denotes the nominal state for the followers. Each component $\tilde{g}_k=[\tilde{p}^{\top}_k, \tilde{\phi}_k]^{\top} \in \mathbb{R}^3$ consists of $\tilde{p}_k=[\tilde{p}_k^x, \tilde{p}_k^y] \in \mathbb{R}^2$ and $\tilde{\phi}_k \in \mathbb{R}$. 

\HT{
\begin{remark}
Unlike rotationally invariant formations (e.g., \cite{Vu2024,Erskine2024,Jing2019,Chen2023,Cao2020,Lin2014,Fang2024,HeGen2025,Aranda2025}), the proposed non-uniform scaling is anisotropic and therefore requires an explicit scaling direction $\theta$ in the nominal configuration $(G, \tilde{g}, \theta)$. This is conceptually analogous to global bearing-based formation \cite{Zhao2016}, where the configuration—though not explicitly containing a $\theta$—is defined in a global coordinate frame that implicitly provides a reference direction. In this paper, $\theta$ is fixed offline. When a mission requires switching $\theta$, the time-varying scaling direction should be acquired by all the agents via communications.
\end{remark}}




The time-varying target state of $(G,\tilde{g}, \theta)$ is parameterized by the stacked vector $g^*(t)=[g_l^{*\top}(t), g_f^{*\top}(t)]^{\top}$, where $g_l^*(t)=[g_1^{*\top}(t), \cdots, g_m^{*\top}(t)]^{\top} \in \mathbb{R}^{3m}$ and $g_f^*(t)=[g_{m+1}^{*\top}(t), \cdots, g_n^{*\top}(t)]^{\top} \in \mathbb{R}^{3(n-m)}$ represent the target state for the leaders and followers, respectively. These target state evolve continuously over time with reference to the nominal configuration $(\tilde{g}, \theta)$. Specifically:
\begin{equation}
g^*(t) = (I_n \otimes S(t, \theta)) \tilde{g} + 1_n \otimes \tau(t),
\label{equ_nc}
\end{equation}
where $t$ is the time variable, $S(t, \theta)=\varTheta \diag(s(t)) \varTheta^{\top}$, $\varTheta=\begin{bmatrix} R(\theta) & 0 \\ 0 & 1 \end{bmatrix}$, $s(t)=[s^{\top}_{p}(t), s_{\phi}(t)]^{\top} \in \mathbb{R}^3$ and $\tau(t)=[\tau^{\top}_{p}(t), \tau_{\phi}(t)]^{\top} \in \mathbb{R}^3$ are time-varying maneuver parameters:
\begin{itemize}
\item $s_p(t) = [s_x(t), s_y(t)]^{\top}$ and $\tau_p(t) = [\tau_x(t), \tau_y(t)]^{\top}$ govern the non‑uniform scaling and translation of the position formation, respectively (see Definition~\ref{de:position_transformation});
\item $s_{\phi}(t)$ and $\tau_{\phi}(t)$ determine the scaling and translation of the attitude formation, respectively (see Definition~\ref{de attitude transformation}).
\end{itemize}

\subsubsection{Sensing Capability}
Each follower agent $k$ is not able to communicate with others, and can only access local relative measurements, including:
(i) the relative positions $\{p_i-p_k\}_{i\in N_k}$ and (ii) relative yaw angles $\{\phi_i-\phi_k\}_{i\in N_k}$.

Each leader agent, functioning as a mobile reference, has the enhanced capability of measuring its absolute state within the global coordinate frame. \HT{In addition, each leader has access to the time‑varying maneuver parameters $s(t)$, $\tau(t)$, as well as its own nominal configuration $(\tilde{g}_l, \theta)$.}

This heterogeneous sensing paradigm aligns with practical scenarios, where leaders may carry high-precision sensors (e.g., IMU-GPS fusion systems \cite{Zhou2025}) while followers rely on vision, UWB or LiDAR for local observations \cite{Mercedes2021,Shin2024, Vrba2025}. 

The distributed non-uniform scaling formation maneuver control problem is then defined as follows.
\HT{
\begin{problem}[{Non-Uniform Scaling Formation Maneuver Control\footnotemark[4]}]\label{p1}
Consider the heterogeneous sensing setup described in this subsection. Given the nominal configuration $(\tilde{g}, \theta)$ and time-varying maneuver parameters $s(t)$, $\tau(t)$. The control objective is to design a distributed controller $(u_k, \omega_k)$ such that all agents, subject to (\ref{equ_dy}), achieve:
	\begin{equation}
	\lim\limits_{t \to \infty} (g_k(t) - g_k^*(t))=0, k \in V,
    \nonumber
	\end{equation}
where $g^*(t)=[\cdots,{g_k^*}^\top(t),\cdots]^{\top}$ is determined by $(\tilde{g}, \theta)$ and the maneuver parameters according to \eqref{equ_nc}. 
\end{problem}}

\footnotetext[4]{\HT{A formation control problem in the literature usually refers to designing a distributed controller such that the agents converge to a desired shape (characterized by a manifold) asymptotically. In this work, to better reflect the advantage of non-uniform scaling transformation, we focus on formation maneuver control, where the agents are required to not only converge to the desired manifold, but also track specific trajectories determined by the leaders. If no leaders exist, the problem will reduce to formation shape control.}}

\section{Maximum Maneuverability and Matrix-Valued Laplacian}\label{Maneuverability}
To solve Problem \ref{p1}, we first analyze the conditions on the nominal configuration that ensure all maneuver parameters are effective. Next, we investigate how to select leaders and design formation rules so that the leaders can fully govern the behavior of the followers, thereby achieving maximal control over the formation (we refer to this system-wide property as maximum maneuverability). Finally, we derive the rank and graph conditions required for maximum maneuverability.
\subsection{Maximum Maneuverability}
In reality, certain nominal configurations can introduce singularities that undermine the effectiveness of maneuver parameters. As shown in Fig.~\ref{fig:scaling}, when all agents are aligned along the $ x $-axis, scaling along the $ y $-axis has no effect on the formation geometry, while $ x $-axis scaling remains effective. In this case, the $ y $-axis scaling parameter $ s_y $ becomes ineffective, resulting in limited maneuverability and inapplicability to complex tasks, such as transitioning from a line to a V-shape. Next, we formalize the concept of a non-singular configuration to address this issue.

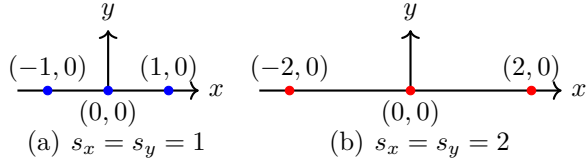
\begin{figure}[t]
    \centering 
    \begin{tikzpicture}[scale=0.8]
    \draw[->, thick] (-1.5,0) -- (1.5,0) node[right]{$x$};
    \draw[->, thick] (0,0) -- (0,1) node[above]{$y$};
    
    \foreach \x in {-1,0,1} {
        \filldraw[blue] (\x,0) circle (2pt);        \ifnum\x=0 
            \node[below] at (\x,0) {$(\x,0)$};
        \else 
            \node[above] at (\x,0) {$(\x,0)$};
        \fi
    }
    \node[below right] at (-1.5,-0.5) {(a) $s_x=s_y=1$};
    
    \begin{scope}[xshift=5cm]
    \draw[->, thick] (-2.5,0) -- (2.5,0) node[right]{$x$};
    \draw[->, thick] (0,0) -- (0,1) node[above]{$y$};
    
    \foreach \x in {-2,0,2} {
        \filldraw[red] (\x,0) circle (2pt);
        \ifnum\x=0 
            \node[below] at (\x,0) {$(\x,0)$};
        \else 
            \node[above] at (\x,0) {$(\x,0)$};
        \fi
    } 
    \node[below right] at (-1.5,-0.5) {(b) $s_x=s_y=2$};
    \end{scope}
    \end{tikzpicture}
    \caption{Singular configuration. (a) Original positions of the agents are aligned with the x-axis. (b) Under scaling transformation with $\theta = 0$, see Equation (\ref{eq_nusp}), a structural singularity occurs: the y-axis scaling parameter $s_y$ becomes ineffective, while x-axis scaling remains effective.} 
    
    \label{fig:scaling}
\end{figure}

From (\ref{equ_nc}), given $\theta$ and $\tilde{g}$, the time-varying target state $g^*(t)$ varies with the maneuver parameters, and all possible states form a space $\varPi(\tilde{g},\theta)$. We term $\varPi(\tilde{g},\theta)$ as the target state space, as defined by the following equation:
\HT{\begin{equation}\label{equ_sapce}
\begin{aligned}
\varPi(\tilde{g},\theta) 
&=\{g \in \mathbb{R}^{3n} : g=(I_n \otimes S(\theta))\tilde{g} + 1_n \otimes \tau, \\
&\qquad S(\theta) = \varTheta \diag(s) \varTheta^{\top},\ s, \tau \in \mathbb{R}^3 \} \\
&=\{g \in \mathbb{R}^{3n} : g_i= \tau + \varTheta \diag(s) \varTheta^{\top} \tilde{g}_i = \tau + \\& \qquad \varTheta \diag(\varTheta^{\top} \tilde{g}_i) s, \ i=1,\cdots,n,\  s, \tau \in \mathbb{R}^3 \} \\
&=\{g \in \mathbb{R}^{3n} : g=\mathcal{F}(s,\tau), \ s, \tau \in \mathbb{R}^3 \}, 
\end{aligned}
\end{equation}}
where $\mathcal{F}(s,\tau) \triangleq A(\tilde{g},\theta) [s^{\top}, \tau^{\top}]^{\top}$, $\tilde{g}_{i,\theta}=\varTheta^{\top} \tilde{g}_i=[\tilde{p}^x_{i,\theta}, \tilde{p}^y_{i,\theta},\tilde{\phi}_i]^{\top}$, 
\begin{equation}\label{eq A}
    A(\tilde{g},\theta)=\begin{bmatrix}
    \varTheta \diag(\tilde{g}_{1,\theta}) & I_3 \\ 
    \vdots & \vdots \\
    \varTheta \diag(\tilde{g}_{n,\theta}) & I_3 \\
    \end{bmatrix} \in \mathbb{R}^{3n \times 6}. 
\end{equation}

\begin{definition}[Non-Singular Configuration]
A nominal configuration $(\tilde{g}, \theta)$ is non-singular if the mapping $\mathcal{F}$ is injective, and is singular otherwise.
\label{defnon-singularConfiguration}
\end{definition}
By the above definition, a non-singular configuration ensures that all maneuver parameters uniquely determine the target state. Next, we establish equivalent conditions for a non-singular configuration.
\begin{lemma}
    \label{lem:equivalence}
    The following statements are equivalent:
    \begin{itemize}
        \item[(a)] $(\tilde{g}, \theta)$ is non-singular;
        \item[(b)] $\rank(A(\tilde{g},\theta)) = 6 $;
        \item[(c)] $\dim(\varPi(\tilde{g}, \theta)) = 6$;
        \item[(d)] each of the sets $\{\tilde{p}^x_{i,\theta}\}_{i \in V}$, $\{\tilde{p}^y_{i,\theta}\}_{i \in V}$, and $\{\tilde{\phi}_i\}_{i \in V}$ is not a singleton.
    \end{itemize}
\end{lemma}


\begin{proof}
    See Appendix \ref{sec:proofLemEquivalence}.
\end{proof}

It is worth noting that the three translation maneuver parameters $\tau_x$, $\tau_y$, and $\tau_{\phi}$ are always effective. In contrast, the effectiveness of the three scaling maneuver parameters requires the validity of the three conditions in Lemma~\ref{lem:equivalence}(d). For example, if the set $\{\tilde{p}^x_{i,\theta}\}_{i \in V}$ is a singleton then the maneuver parameter $s_x$ becomes ineffective.

To enable formation maneuver control with robust adaptability to complex environments and diverse mission requirements, maintaining maximum maneuverability is essential. Under the leader-follower strategy, a singular nominal configuration of the leaders compromises formation maneuverability by rendering certain maneuver parameters ineffective. 

Moreover, even if the leader configuration is non-singular, followers constrained by local sensing may still fail to track leader state changes. This highlights the challenge of ensuring that the influence of leader motions can fully and uniquely propagate throughout the formation.
\HT{
Inspired by \cite{Li2018,Zhao2018,yang2019,Hector2021,Morbidi2022,Chen20231,Fang2024,FANG2025,Huang2025}, we seek a Laplacian $M=[M_{ki}] \in\mathbb{R}^{3n\times 3n}$ determined by $(G,\tilde{g}, \theta)$ such that 
\begin{equation}\label{eq_mg}
\varPi(\tilde{g},\theta)=\{g\in\mathbb{R}^{3n}:Mg = 0\},    
\end{equation}
where $g = [g_l^\top, g_f^\top]^\top$ represents the combined state of leaders and followers, and $ M_{ki} $ reflects the interaction weight between agent $k$ and agent $i$. Thus, any feasible formation state $g$ satisfies $g \in \nullspace(M) = \varPi(\tilde{g},\theta)$. If, for any given leader state $g_l$, the follower state $g_f$ is uniquely determined by solving $M [g_l^\top, g_f^\top]^\top = 0$, then any change of $g_l$ induces a corresponding change in $g_f$.

Now, we formally define maximum maneuverability in the leader-follower framework as follows.
\begin{definition}[Maximum Maneuverability]\label{deMaximumManeuverability}
    A nominal formation  $(G,\tilde{g}, \theta)$ in $\mathbb{R}^2$ is said to achieve maximum maneuverability under the leader-follower strategy with Laplacian $M$ satisfying \eqref{eq_mg} if
    \begin{itemize}
        \item[(a)] the leaders' nominal configuration $(\tilde{g}_l, \theta)$ is non-singular;
        \item[(b)] for any given leader state $g_l$, the follower state $g_f$ is uniquely determined by solving $M [g_l^\top, g_f^\top]^\top = 0$.
    \end{itemize}
\end{definition}}
To achieve maximum maneuverability, we next address three key design aspects: leader selection criteria, construction of the matrix‑valued Laplacian, and the characterization of maximum maneuverability through both algebraic conditions and graph‑theoretic requirements.

\subsection{Leader Selection for Maximum Maneuverability}

The following lemma further gives equivalent conditions for the convenience of leader selection.
\begin{lemma} \label{lemma_leader_selection}
    The leaders' nominal configuration $(\tilde{g}_l, \theta)$ is non-singular if and only if the following conditions both hold:
    \begin{itemize}
        \item[(a)] the number of leaders satisfies $m \ge 2$;
        \item[(b)] each of the sets $\{\tilde{p}^x_{i,\theta}\}_{i \in V_l}$, $\{\tilde{p}^y_{i,\theta}\}_{i \in V_l}$, and $\{\tilde{\phi}_i\}_{i \in V_l}$ is not a singleton.
    \end{itemize}
\end{lemma}

\begin{proof}
    According to Lemma \ref{lem:equivalence}, the result follows directly.
\end{proof}

When the leader nominal configuration is non-singular, there exists a one-to-one correspondence between the leader state and the maneuver parameters. Next we show how to explicitly compute these maneuver parameters. 

From (\ref{equ_sapce}), we have $g_l=A(\tilde{g}_l,\theta) z$, where $z=[s^{\top}, \tau^{\top}]^{\top}$, and $A(\tilde{g}_l,\theta) \in \mathbb{R}^{3m \times 6}$ is obtained by substituting $\tilde{g}$ in Equation (\ref{eq A}) with $\tilde{g}_l$. Lemma~\ref{lemma_leader_selection} implies that $\rank(A(\tilde{g}_l, \theta))=6$. As a result, the maneuver parameters can be uniquely determined as:
\begin{equation}
    z=\left(A^{\top}(\tilde{g}_l,\theta) A(\tilde{g}_l,\theta) \right)^{-1} A^{\top}(\tilde{g}_l,\theta) g_l.
    \nonumber
\end{equation}
\subsection{Matrix-Valued Laplacian for Maximum Maneuverability}\label{ss_mvlfmm}

To construct a Laplacian matrix satisfying Definition \ref{deMaximumManeuverability}, we firstly introduce the following matrix-valued constraint involving three agents $i, j, k$:
\begin{equation}
    W_{jk}(\tilde{g}_{ijk},{\theta}) g_{ik}+W_{ki}(\tilde{g}_{ijk},{\theta}) g_{jk}=0,
    \label{equ_vvc}
\end{equation}
\HT{where $g_{ik}=g_i-g_k$, $g_{jk}=g_j-g_k$, $\tilde{g}_{ijk}=[\tilde{g}_i^{\top}, \tilde{g}_j^{\top}, \tilde{g}_k^{\top}]^{\top}$, $\tilde{\phi}_{jk} = \tilde{\phi}_j - \tilde{\phi}_k$, $\tilde{\phi}_{ki} = \tilde{\phi}_k - \tilde{\phi}_i$, $\tilde{p}_{jk,\theta} = R^\top(\theta) \tilde{p}_{jk}$, $\tilde{p}_{ki,\theta} = R^\top(\theta) \tilde{p}_{ki}$, $\tilde{p}_{jk} = \tilde{p}_j - \tilde{p}_k$, $\tilde{p}_{ki} = \tilde{p}_k - \tilde{p}_i$, 
\begin{align*}
W_{jk}(\tilde{g}_{ijk},{\theta}) &= w_{jk} \varTheta^{\top}, &
w_{jk} &= \begin{bmatrix} \diag\left(\tilde{p}_{jk,\theta}\right) & 0 \\ 0 & \tilde{\phi}_{jk} \end{bmatrix},\\
W_{ki}(\tilde{g}_{ijk},{\theta}) &= w_{ki} \varTheta^{\top}, & w_{ki} &= \begin{bmatrix} \diag\left(\tilde{p}_{ki,\theta}\right) & 0 \\ 0 & \tilde{\phi}_{ki} \end{bmatrix}.
\end{align*}
The weight matrices $W_{jk},W_{ki}$ are constructed from the pre-designed nominal state vector $\tilde{g}_{ijk}$ and the scaling direction $\theta$. Constraint \eqref{equ_vvc} requires that agent $k$ can sense both agents $i$ and $j$ (i.e., $(i,k),(j,k)\in E$).


Intuitively, the use of two edges is necessary because a single-edge constraint $W_{ik}g_{ik}=0$ would force $W_{ik}=0$ to maintain invariance under non-uniform scaling. Here, $\varTheta^\top$ projects the relative states onto a new coordinate frame whose $x$-axis aligns with the scaling direction $\theta$. Within this frame, the diagonal weight matrices $w_{jk}$ and $w_{ki}$ apply independent scaling along each axis. Crucially, the constraint possesses a key invariance property: it remains satisfied if $g_i,g_j,g_k$ undergo the same combined translation and non-uniform scaling transformation (see Lemma~\ref{lem:2}). This invariance provides a mechanism for coordinated maneuvers: by controlling leaders to follow a desired non-uniform scaling transformation along $\theta$, followers that satisfy \eqref{equ_vvc} will automatically undergo the same transformation, enabling formation-wide non-uniform scaling maneuvers. Extension of constraint \eqref{equ_vvc} to $\mathbb{R}^d$ ($d>2$) is conceptually possible by redefining the position components as $\tilde{p}_i \in\mathbb{R}^d$ and appropriately representing attitudes.

Taking Fig.~\ref{fig:figemvc} as an example, the nominal state of agents $i,j,k$ are $\tilde{g}_i=[-2, 1, \pi/2]^{\top}$, $\tilde{g}_j=[1.5,0.5,\pi/4]^{\top}$, $\tilde{g}_k=[0,0 ,0]^{\top}$, respectively, and $\theta=0$.} Then,
\begin{equation}
\nonumber
W_{jk}=\begin{bmatrix} 1.5 & 0 & 0 \\ 0 & 0.5 & 0 \\ 0 & 0 & \pi/4 \\ \end{bmatrix}, W_{ki}=\begin{bmatrix} 2 & 0 & 0 \\ 0 & -1 & 0 \\ 0 & 0 & -\pi/2 \\ \end{bmatrix}.
\end{equation}

The weight matrices apply a non-uniform scaling transformation to the relative state vector, ensuring that the sum of two directed edges under this transformation equals zero. Note that the choice of weight matrices is not unique.
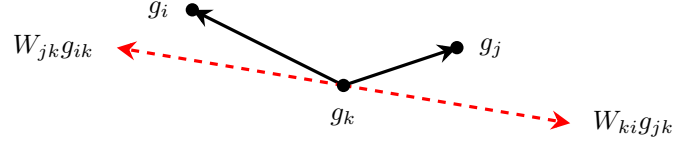
\begin{figure}[!htb]
	\centering
	\begin{tikzpicture}
	
	\draw (-2,1) node[left=5pt] {$g_i$};
	\draw (1.5,0.5) node[right=5pt] {$g_j$};
	\draw (0,0) node[below=5pt] {$g_k$};
	
	\draw (-3,0.5) node[left=5pt] {$W_{jk}g_{ik}$};
	\draw (3,-0.5) node[right=5pt] {$W_{ki}g_{jk}$};
	
	\begin{scope}[very thick, every node/.style={sloped,allow upside down}]
	\draw[-{Stealth[length=2.5mm,width=2.5mm]}] (0,0)-- node {} (-2,1);
	\draw[-{Stealth[length=2.5mm,width=2.5mm]}] (0,0)-- node {} (1.5,0.5);
	\draw[dashed, -{Stealth[length=2.5mm,width=2.5mm]}, color=red] (0,0)-- node {} (-3,0.5);
	\draw[dashed, -{Stealth[length=2.5mm,width=2.5mm]}, color=red] (0,0)-- node {} (3,-0.5);
	\end{scope}

	\fill (-2,1) circle (2.5pt);
	\fill (1.5,0.5) circle (2.5pt);
	\fill (0,0) circle (2.5pt);
	
	\end{tikzpicture} 
	
	\caption{An example of the matrix-valued constraint.}     
	\label{fig:figemvc}      
\end{figure}

We now present a key property of the matrix-valued constraint.
\begin{lemma} \label{lem:2}
    The constraint (\ref{equ_vvc}) is invariant to translation and non-uniform scaling transformation of $g_i, g_j, g_k$.
\end{lemma}
\begin{proof}
    For each $l \in \{i,j,k\}$, we apply a translation $\tau$ and a non-uniform scaling transformation $S$, obtaining:
    \begin{equation}
    \begin{split}
    g'_l = S g_l + \tau, \quad l \in \{i,j,k\}.
    \end{split}
    \nonumber
    \end{equation}
    
    We now demonstrate that the transformed vectors $g'_i, g'_j, g'_k$ satisfy the matrix-valued constraint in (\ref{equ_vvc}):
    \begin{equation}
    \begin{split}
    &W_{jk} g'_{ik}+W_{ki} g'_{jk}
    =W_{jk} S g_{ik}+W_{ki} S g_{jk}\\
    &=\begin{bmatrix} \diag\left(\tilde{p}_{jk, \theta}\right)  & 0 \\  0 & \tilde{\phi}_{jk} \end{bmatrix} \diag(s) \varTheta^{\top} g_{ik}\\ &~~ + \begin{bmatrix} \diag\left(\tilde{p}_{ki,\theta}\right)  & 0 \\  0 & \tilde{\phi}_{ki} \end{bmatrix} \diag(s) \varTheta^{\top}  g_{jk}\\
    &=\diag(s) \left( W_{jk} g_{ik}+W_{ki} g_{jk} \right) =0,
    \end{split}
    \nonumber
    \end{equation}
    the final equality holds because multiplication of diagonal matrices is commutative.
\end{proof}

Next, we construct a matrix-valued Laplacian based on the proposed matrix-valued constraint. Let the constraint index set be defined as $C = \left\{(i,j,k) \in V^3 : (i,k),(j,k) \in E, i < j \right\}$. The set of all constraints associated with the sensing graph $G$ is then given by $\{ W_{jk} g_{ik}+W_{ki} g_{jk}=0 : (i,j,k) \in C \}$.

Each matrix-valued constraint defined in (\ref{equ_vvc}) corresponding to the constraint index $(i,j,k) \in C$ can be aggregated into a matrix-valued Laplacian $M(G,\tilde{g}, \theta) \in \mathbb{R}^{3n \times 3n}$ satisfying:
\begin{equation}\label{equ_M}
M g =0,
\end{equation}
where the matrix block located at the $k$th row and $i$th column of $M$, denoted $M_{ki}$, is defined as follows:

\begin{equation}\label{equ_mvl}
    M_{ki} =
    \begin{cases}
        \sum\limits_{(i,j,k) \in C} W_{jk} + \sum\limits_{(j',i,k) \in C} W_{kj'}, & \text{if } k \ne i, \\
        \sum\limits_{(i,j,k) \in C} W_{ij} + \sum\limits_{(j',i,k) \in C} W_{j'i}, & \text{if } k = i.
    \end{cases}
\end{equation}
By convention, each summation is defined to be zero when the corresponding index set is empty.





We now investigate key properties of the matrix-valued Laplacian $M(G,\tilde{g}, \theta)$.
\HT{
\begin{lemma} \label{lem:LaplacianProperties}
    For any nominal configuration $(\tilde{g}, \theta)$, it always holds that $\varPi(\tilde{g},\theta) \subseteq \nullspace(M(G,\tilde{g}, \theta))$. Moreover, if $(\tilde{g}, \theta)$ is non-singular, then the following conditions are equivalent:
    \begin{itemize}
        \item[(a)] $\nullspace(M(G,\tilde{g}, \theta))=\varPi(\tilde{g},\theta)$,
        \item[(b)] $\rank(M(G,\tilde{g}, \theta)) = 3n-6$.
    \end{itemize}
\end{lemma}
\begin{proof}
    $\varPi(\tilde{g},\theta) \subseteq \nullspace(M(G,\tilde{g}, \theta))$ follows directly from the definition of $M(G,\tilde{g}, \theta)$ and Lemma \ref{lem:2}. For non-singular $(\tilde{g}, \theta)$, Lemma \ref{lem:equivalence} gives $\dim(\varPi(\tilde{g},\theta)) = 6$. Since $\varPi(\tilde{g}, \theta) \subseteq \nullspace(M(G,\tilde{g}, \theta))$, $\nullspace(M(G,\tilde{g}, \theta)) = \varPi(\tilde{g}, \theta)$ if and only if $\dim(\nullspace(M(G,\tilde{g}, \theta))) = \dim(\varPi(\tilde{g}, \theta))$. Given that $\dim(\varPi(\tilde{g}, \theta)) = 6$, it follows that $\nullspace(M(G,\tilde{g}, \theta)) = \varPi(\tilde{g}, \theta)$ if and only if $\rank(M(G,\tilde{g}, \theta)) = 3n - 6$.
\end{proof}
}
To this point, we have derived the condition on the matrix-valued Laplacian for characterizing the target configuration space $\varPi(\tilde{g}, \theta)$. Next, we investigate how the leader state uniquely determine the follower state through the matrix-valued Laplacian. We begin by reformulating \eqref{equ_M} as
\begin{equation}
M g = \hat{M} \diag(\varTheta^{\top}) g =0,
\nonumber
\end{equation}
and subsequently partition $\hat{M}$ based on the leader-follower structure to facilitate this analysis.
\begin{equation}
\hat{M}= \begin{bmatrix}
    \hat{M}_l \\
    \hat{M}_f
\end{bmatrix}=\left[ \begin{array}{cc}
\hat{M}_{ll} & \hat{M}_{lf}\\
\hat{M}_{fl} & \hat{M}_{ff} \end{array} \right],
\nonumber
\end{equation}
where $\hat{M}_l = [\hat{M}_{ll} \, \hat{M}_{lf}]$, $\hat{M}_f = [\hat{M}_{fl} \, \hat{M}_{ff}]$, $\hat{M}_{ll} \in \mathbb{R}^{3m \times 3m}$, $\hat{M}_{lf} \in \mathbb{R}^{3m \times 3(n-m)}$, $\hat{M}_{fl} \in \mathbb{R}^{3(n-m) \times 3m}$ and $\hat{M}_{ff} \in \mathbb{R}^{3(n-m) \times 3(n-m)}$. Based on this partitioning, we obtain
\begin{equation} \label{eq follower}
    \hat{M}_{fl} \diag(\varTheta^{\top})  g_l + \hat{M}_{ff} \diag(\varTheta^{\top}) g_f = 0.
\end{equation}

If the block matrix $\hat{M}_{ff}$ is non-singular, the follower state $g_f$ can be uniquely determined by
\begin{equation}\label{eq_gf}
g_f = - \diag(\varTheta) \hat{M}_{ff}^{-1} \hat{M}_{fl} \diag(\varTheta^{\top}) g_l.
\end{equation}

\HT{The foregoing analysis leads to the following algebraic characterization of maximum maneuverability.
\begin{theorem} \label{thm:max_maneuver_char}
    A nominal formation $(G,\tilde{g}, \theta)$ in $\mathbb{R}^2$ achieves maximum maneuverability under leader-follower strategy if and only if the following conditions are satisfied:
    \begin{itemize}
        \item[(a)] The number of leaders satisfies $m \ge 2$;
        \item[(b)] Each of the sets $\{\tilde{p}^x_{i,\theta}\}_{i \in V_l}$, $\{\tilde{p}^y_{i,\theta}\}_{i \in V_l}$, and $\{\tilde{\phi}_i\}_{i \in V_l}$ is not a singleton;
        \item[(c)] The block matrix $\hat{M}_{ff}$ in \eqref{eq follower} is non-singular.
    \end{itemize}
\end{theorem}
\begin{proof}
Conditions (a)–(b) are equivalent to Definition~\ref{deMaximumManeuverability}(a) by Lemma~\ref{lemma_leader_selection}. Condition (c) is equivalent to Definition~\ref{deMaximumManeuverability}(b) because $\hat{M}_{ff}$ being non-singular ensures a unique solution for $g_f$ (see \eqref{eq_gf}).
\end{proof}}



\subsection{Sensing Graphs for Maximum Maneuverability}
\label{subsec:sensing_graphs_condition}


\HT{
According to Theorem~\ref{thm:max_maneuver_char}, non-singularity of the block matrix $\hat{M}_{ff}$ in \eqref{eq follower} is essential for $(G, \tilde{g}, \theta)$ to achieve maximum maneuverability. However, verifying this algebraic condition requires centralized global computation and is sensitive to specific state values. To obtain more design-friendly conditions, we establish necessary and sufficient graph-theoretic requirements ensuring that $\hat{M}_{ff}$ is non-singular.

The key insight is that for a 2‑rooted graph $G$ (which, by Lemma~\ref{lem:mp2rg_equivalence}, contains a spanning DEP‑induced graph $\mathcal{L}_{\kappa}$), the non‑singularity of $\hat{M}_{ff}$ can be analyzed DEP‑by‑DEP. The matrix $\hat{M}_{ff}$ fails to be invertible only when the configuration $(\tilde{g}, \theta)$ belongs to a certain algebraic variety of measure zero. This degenerate set can be explicitly described by the following condition on each DEP.}

\begin{assumption} \label{ass_2rootrank}
    Consider a nominal formation $(G, \tilde{g}, \theta)$ in $\mathbb{R}^2$, where $G$ is a 2-rooted graph with a spanning DEP-induced graph $\mathcal{L}_{\kappa}$, each DEP $G_{\mathcal{P}_h} = (V_{\mathcal{P}_h}, E_{\mathcal{P}_h})$, $h = 1, \dots, \kappa$, with entry agents $\{i_h, j_h\}$, satisfies:
   
    
    \begin{equation}\label{eq nondegeneracy condition}
        \prod_{\{u, v\} \in \mathcal{E}_{\mathcal{P}_h}} \tilde{p}^x_{u v,\theta} \tilde{p}^y_{u v,\theta} \tilde{\phi}_{u v} \ne 0,
    \end{equation}
    where $\mathcal{E}_{\mathcal{P}_h} = \left\lbrace \{u, v\} : (u, v), (v, u) \in E_{\mathcal{P}_h} \right\rbrace  \cup \{ i_h, j_h \}$, $\tilde{p}^x_{uv,\theta} = R^\top(\theta) \tilde{p}^x_{uv}$, $\tilde{p}^y_{uv,\theta} = R^\top(\theta) \tilde{p}^y_{uv}$, and $\tilde{\phi}_{uv} = \tilde{\phi}_u - \tilde{\phi}_v$.
\end{assumption}

\HT{\begin{remark}\label{rem:design_feasibility}
Condition \eqref{eq nondegeneracy condition} can be satisfied via a simple design procedure. For the position part, given distinct nominal positions $\{\tilde{p}_i\}$, compute edge direction angles $\psi_{uv} = \operatorname{atan2}(\tilde{p}^y_{uv}, \tilde{p}^x_{uv}) \pmod{2\pi}$. The forbidden $\theta$ values form a finite set $\{(\psi_{uv} - k\pi/2) \pmod{2\pi} \mid \{u,v\}\in \mathcal{E}, \ k=0,1,2,3\}$ (at most $4|\mathcal{E}|$ points), where $\mathcal{E} = \bigcup_{h=1}^{\kappa} \mathcal{E}_{\mathcal{P}_h}$. Select $\theta$ as the midpoint of the largest gap between sorted forbidden values. For the attitude part, the condition is easily ensured by assigning distinct nominal yaw angles to all agents.
\end{remark}
\begin{remark}\label{rem:design_feasibility1}
Compared to affine \cite{Zhao2018}, distance-ratio-based \cite{Cao2020}, clique-based \cite{HeGen2025}, angle‑based \cite{Buckley2021}, and bearing‑based \cite{Erskine2024} position formation control approaches—which impose non‑collinearity restrictions—our assumption merely requires distinct positions $\{\tilde{p}_i\}$, a substantially milder condition. For attitudes, although the nominal configuration requires distinct headings, the actual formation can achieve uniform alignment ($s_\phi=0$) or scaled relative headings ($s_\phi\neq0$), thereby subsuming both consensus \cite{KwangKyo2014} and fixed‑relative‑attitude schemes \cite{Dimarogonas2009,Song2017} while providing enhanced maneuverability.
\end{remark}

We can now state a concise graph‑theoretic characterization of maximum maneuverability.}

\begin{theorem} \label{the:rankL}
    A nominal formation $(G,\tilde{g}, \theta)$ in $\mathbb{R}^2$ achieves maximum maneuverability under a leader-follower strategy if and only if $G$ is 2-rooted with the two roots as leaders, and the nominal formation $(G, \tilde{g}, \theta)$ satisfies Assumption \ref{ass_2rootrank}.
\end{theorem}

\begin{proof}
    See Appendix \ref{sec:proofTheRankL}.
\end{proof}

\HT{
\begin{remark}\label{rem:sparsity_comparison}
While some existing approaches (e.g., \cite{ZHANG2025}) rely solely on rank conditions of the Laplacian, we provide explicit graph-theoretic conditions that are easier to verify and design. Compared to affine formation control \cite{Zhao2018,Lin2016} which requires $3$-rooted graphs in $\mathbb{R}^2$ (and thus at least $3$ neighbors for each follower and $3$ leaders), our 2-rooted condition permits sparser graphs (e.g., in a DEP structure each inner agent has exactly two neighbors) and requires only two leaders. Unlike methods that only give graph conditions \cite{Lin2014}, we additionally provide the explicit non-degeneracy condition \eqref{eq nondegeneracy condition}, which can be directly verified or enforced during formation design.
\end{remark}
}

\section{Formation Maneuver Control}\label{Control}

\HT{
Based on the preceding analysis, particularly the maximum maneuverability property established in Theorem \ref{the:rankL}, we propose the designed distributed controller in this section. The maximum maneuverability guarantees two critical prerequisites for controller design: (i) the block matrix $\hat{M}_{ff}$ is non-singular, enabling the follower state to be uniquely determined by the leader state—a property fundamental for constructing a well-defined feedback law; and (ii) all six maneuver parameters are effective, providing full control authority over the formation's geometry. 

However, to ensure the stability and implementability of the specific control laws, additional technical conditions are required. These conditions can be satisfied by following a design approach similar to that outlined in Remark~\ref{rem:design_feasibility}.


\begin{assumption} \label{ass_2rootstab}
    Consider a nominal formation $(G, \tilde{g}, \theta)$ in $\mathbb{R}^2$, where $G$ is a 2-rooted graph with a spanning DEP-induced graph $\mathcal{L}_{\kappa}$. Each DEP $G_{\mathcal{P}_h} = (V_{\mathcal{P}_h}, E_{\mathcal{P}_h})$, $h = 1, \dots, \kappa$ , with entry agents $\{i_h, j_h\}$ and inner agents $\{1, \dots, \ell_h\}$, satisfies:
    \begin{equation} \label{eq stabilitycondition}
        \prod_{l=2}^{\ell_h} \tilde{p}^x_{i_hl,\theta} \tilde{p}^y_{i_hl,\theta} \tilde{\phi}_{i_hl} \neq 0.
    \end{equation}
\end{assumption}

\begin{assumption} \label{ass_c}
   For each matrix-valued constraint $(i,j,k) \in C$ defined in (\ref{equ_vvc}), the nominal configuration $(\tilde{g},\theta)$ satisfies
   \begin{equation}
       \tilde{p}^x_{ji, \theta} \ne 0, \quad \tilde{p}^y_{ji, \theta} \ne 0, \quad \text{and} \quad \tilde{\phi}_{ji} \ne 0.
       \nonumber
   \end{equation}
\end{assumption}}





\subsection{Distributed Formation Maneuver Control Laws}\label{ss_dfmcl}
In this subsection, we propose distributed non-uniform scaling formation maneuver control laws, in the scenarios with stationary leaders and moving leaders, respectively.

\HT{
According to the control objective described in Problem \ref{p1}, we define the tracking errors for followers and leaders as $\delta_l = g_l - g_l^*$ and $\delta_f = g_f - g_f^*$, respectively, where the target follower state is given by
\begin{equation}
    g_f^* = -\left(\hat{M}_{ff} \diag(\varTheta^{\top}) \right)^{-1} \hat{M}_{fl} \diag(\varTheta^{\top}) g^*_l,
    \nonumber
\end{equation}
which follows directly from \eqref{eq_gf}. The invertibility of $\hat{M}_{ff}$ in this expression is guaranteed by the maximum maneuverability conditions of Theorem \ref{the:rankL}.}

The control objective is thus reformulated as designing a distributed control law such that $ \delta_f \to 0 $ and $ \delta_l \to 0 $ as $ t \to \infty $.

\subsubsection{Stationary Leaders} We first consider the case where leaders are stationary, i.e., $g_l=g^*_l$ and $\dot{g}^*_l=0$. In this case, the closed‑loop control law can be written in compact form as
\begin{equation}\label{equ_sluf}
    \begin{cases}
    \dot{g}_l=0, \\
    \dot{g}_f = -\diag(\varTheta) D \hat{M}_{ff} \diag(\varTheta^{\top}) \delta_f,
    \end{cases}
\end{equation}
where $D = \diag(D_k)$ is a diagonal matrix to be designed to ensure the convergence of the tracking error, and each $D_k \in \mathbb{R}^{3 \times 3}$ is a non-zero diagonal gain matrix corresponding to agent $k$. \HT{The existence of such a diagonal stabilizing matrix $D$ is guaranteed under Assumption~\ref{ass_2rootrank} and \ref{ass_2rootstab}.}

\HT{
According to \eqref{equ_mvl}, the closed‑loop dynamics for each follower reduce to
\begin{equation}\label{uk}
    \dot{g}_k =\begin{bmatrix} u_k \\
\omega_k \end{bmatrix} = -\varTheta D_k \sum_{(i,j,k) \in C} \left( W_{jk} g_{ik} + W_{ki} g_{jk} \right), k \in V_f.
\end{equation}

The explicit form of controller \eqref{uk} requires each follower to measure relative states from at least two neighbors. This structural requirement is ensured by the 2‑rooted graph condition for maximum maneuverability.}

\begin{theorem}\label{th convergence}
   Let the nominal formation $(G,\tilde{g}, \theta)$ in $\mathbb{R}^2$ satisfy Assumptions \ref{ass_2rootrank} and \ref{ass_2rootstab}. There exists a diagonal matrix $D$ such that the tracking error $\delta_f$ converges to zero globally and exponentially fast under the control law (\ref{equ_sluf}).
\end{theorem}

\begin{proof}
    See Appendix \ref{sec:proofTheconvergence}.
\end{proof}
\HT{
\begin{remark}
In \cite{Lin2014}, a stabilizing matrix exists for almost all Laplacians whose kernel contains the nominal configuration, but infeasible cases are not clearly identified. By contrast, Assumption~\ref{ass_2rootstab} provides an explicit condition on the nominal configuration, guaranteeing a stabilizing matrix always exists when the Laplacian follows \eqref{equ_mvl}.
\end{remark}}

\subsubsection{Moving Leaders} To address moving leaders with time-varying velocities, we propose the following closed‑loop control law that utilizes absolute velocity feedback, similar to the approach in \cite{Zhao2018,Fang2022}.
\HT{
\begin{equation}\label{equ_mluf}
\begin{split}
\dot{g}_k &=\begin{bmatrix} u_k \\ \omega_k \end{bmatrix}  \\
   &=\begin{cases}
    -k_l (g_k - g_k^*) + \dot{g}_k^*, & \text{$k \in V_l$,} \\
     W_{kk}^{-1} [W_{jk}(k_fg_{ik} + \dot{g}_i)+ W_{ki}(k_fg_{jk} + \dot{g}_j)], & \text{$k \in V_f$,}
    \end{cases}
\end{split}
\end{equation}
where $W_{kk} = W_{jk} + W_{ki}$, and $k_l, k_f > 0$ are control gains. The invertibility of $W_{kk}$ is ensured by Assumption~\ref{ass_c}.}



\begin{theorem} \label{the_dd}
    Under Assumptions \ref{ass_2rootrank} and \ref{ass_c}. If the leader velocity $\dot{g}_l^*(t)$ is time-varying and continuous, then the tracking errors $\delta_l$ and $\delta_f$ of the single-integrator multi-agent systems converge to zero globally and exponentially fast under the control law (\ref{equ_mluf}).
\end{theorem}

\begin{proof}
    Under Assumption \ref{ass_c}, $W_{kk}$ is non-singular. The matrix-vector form of (\ref{equ_mluf}) is
    \begin{equation}
    \begin{cases}
    \dot{\delta}_l =- k_l\delta_l, \\
    \hat{M}_{ff} \diag(\varTheta^{\top}) \dot{\delta}_f = - \hat{M}_{ff} \diag(\varTheta^{\top}) k_f\delta_f.
    \end{cases}
    \nonumber
    \end{equation}
    Since $\diag(\varTheta^{\top})$ is non-singular and $G$ satisfies Assumption \ref{ass_2rootrank}, according to Theorem \ref{the:rankL}, we know that $\hat{M}_{ff}$ is non-singular, then we have $\dot{\delta}_f = - k_f\delta_f$. Additionally, $\dot{\delta}_l =- k_l\delta_l$, which implies that $\delta_l$ and $\delta_f$ globally converge to zero at an exponential rate.
\end{proof}

\HT{
\begin{remark}
    Similar to Assumption 2 in \cite{Fang2024}, our Assumption~\ref{ass_c} ensures the non-singularity of the matrix $W_{kk}$. While \cite{Zhao2018} requires the Laplacian to be positive semi-definite and satisfy a rank condition for this property, our approach instead imposes only a rank condition on the Laplacian, making the assumption substantially weaker.
\end{remark}}

\subsection{Design of the Diagonal Stabilizing Matrix $ D $}\label{sec_design_D}
In the preceding section, we have obtained the global convergence of the proposed controller based on the existence of $D$. 
However, computing $ D $ remains a challenging inverse eigenvalue problem, which can be formulated as
\begin{equation}\label{computing D}
\begin{array}{ll}
   \textbf{find}  &  D = \diag(x)\\
    \textbf{subject to} &  \forall \lambda \in \sigma(D\hat{M}_{ff}), \\ 
                  &  \Re(\lambda) > 0,
\end{array}
\end{equation}
where $x \in \mathbb{R}^{3(n-m)}$, and $\sigma(\cdot)$ denotes the matrix spectrum. 

This problem is inherently nonlinear, non-convex, and high-dimensional, typically requiring centralized computation \cite{Lin2014,Lin2016,Yang2021}. Nevertheless, For each DEP with $\ell =1$ or $2$, we obtain explicit closed-form expressions.
\begin{theorem} \label{the_dsm}
Under Assumptions \ref{ass_2rootrank} and \ref{ass_2rootstab}, for a DEP-induced graph $\mathcal{L}_{\kappa}$. Each DEP $G_{\mathcal{P}_h} = (V_{\mathcal{P}_h}, E_{\mathcal{P}_h})$, $h = 1, \dots, \kappa$ , with $\ell_h \in \{1, 2\}$ inner vertices, satisfies:

\begin{itemize}
    \item If $\ell_h = 1$, with $V_{\mathcal{P}_h} = \{i, j, k\}$, $E_{\mathcal{P}_h} = \{(i, k), (j, k)\}$, then $D^h = w_{ij}^{\top}$, and $D^h \hat{M}^h_{ff}$ has positive eigenvalues.
    
    \item If $\ell_h = 2$, with $V_{\mathcal{P}_h} = \{i, j, k, l\}$, $E_{\mathcal{P}_h} = \{(i, k), (k, l), (l, k), (j, l)\}$, then
    \[
    D^h = \begin{bmatrix}
        \sgn(w_{il}) |w_{kl} w_{ij} w_{kj}| + w_{il} & 0 \\
        0 & w_{kl} w_{ij} w_{il}
    \end{bmatrix}
    \]
    and $D^h \hat{M}^h_{ff}$ has positive eigenvalues.
\end{itemize}
Thus, there exists a diagonal matrix $D = \diag( D^h )$ such that $D \hat{M}_{ff}$ has positive eigenvalues, where $w_{ij}$ is a diagonal matrix defined in \eqref{equ_vvc} and $\hat{M}_{ff}$, $\hat{M}^h_{ff}$ are defined in \eqref{equ_mf}.
\end{theorem}

\begin{proof}
Case $\ell_h = 1$:
     $\hat{M}^h_f$ for this configuration is
    \begin{equation}
    \hat{M}^h_f =\left[ \hat{M}^h_{fl} \, \hat{M}^h_{ff} \right]=\left[ \begin{array}{cc|c}
    w_{jk} & w_{ki} & w_{ij}
    \end{array} \right]. 
    \nonumber
    \end{equation}
    Under Assumption \ref{ass_2rootrank} and \ref{ass_2rootstab}, we obtain
    \begin{equation} 
    D^h \hat{M}^h_{ff} = w_{ij}^{\top} w_{ij},
    \nonumber
    \end{equation}
    which is a positive definite matrix. Thus, all eigenvalues of $ D^h \hat{M}^h_{ff} $ are strictly positive.

Case $\ell_h = 2$:
    $\hat{M}^h_f$ for this configuration is:
    \begin{equation}
    \hat{M}^h_f =\left[ \hat{M}^h_{fl} \, \hat{M}^h_{ff} \right]=\left[ 
    \begin{array}{cc|cc}
    w_{lk} & 0 & w_{il} & w_{ki} \\
    0 & w_{lk} & w_{jl} & w_{kj}
    \end{array} \right] .
    \nonumber
    \end{equation}
    Apply the permutation $ P = [e_1, e_4, e_2, e_5, e_3, e_6] $, where $ e_i $ are standard basis vectors, to get $ M' = P^T (D^h \hat{M}^h_{ff}) P = \diag(A_x, A_y, A_\phi) $,
    \begin{equation}
        \begin{split}
            d_q &= \sgn(\tilde{p}^q_{il,\theta}) |\tilde{p}^q_{kl,\theta} \tilde{p}^q_{ij,\theta} \tilde{p}^q_{kj,\theta}| + \tilde{p}^q_{il,\theta}, \\
            A_q &= \begin{bmatrix}
            d_q \tilde{p}^q_{il,\theta} & d_q \tilde{p}^q_{ki,\theta} \\
            \tilde{p}^q_{kl,\theta} \tilde{p}^q_{ij,\theta} \tilde{p}^q_{il,\theta} \tilde{p}^q_{jl,\theta} & \tilde{p}^q_{kl,\theta} \tilde{p}^q_{ij,\theta} \tilde{p}^q_{il,\theta} \tilde{p}^q_{kj,\theta}
            \end{bmatrix},
        \end{split}
        \nonumber
    \end{equation}
    where $q \in \{x, y, \phi\}$, and for $q = \phi$, $\tilde{p}^q_{ij,\theta} = \tilde{\phi}_{ij}$ (i.e., the $\theta$ subscript is omitted).
    
    For $ A_\phi $, under Assumption \ref{ass_2rootrank} and \ref{ass_2rootstab}, the trace is:
    \begin{equation}
    \tr(A_\phi) = |\tilde{\phi}_{kl} \tilde{\phi}_{ij} \tilde{\phi}_{kj}| |\tilde{\phi}_{il}| + \tilde{\phi}_{il}^2 + \tilde{\phi}_{kl} \tilde{\phi}_{ij} \tilde{\phi}_{il} \tilde{\phi}_{kj} > 0.
    \nonumber
    \end{equation}
    Since $\tilde{\phi}_{il} \tilde{\phi}_{kj} - \tilde{\phi}_{jl}\tilde{\phi}_{ki} = \tilde{\phi}_{il} \tilde{\phi}_{kj} - \tilde{\phi}_{jl}(\tilde{\phi}_{kj}+\tilde{\phi}_{jl}+\tilde{\phi}_{li}) = \tilde{\phi}_{ij} \tilde{\phi}_{kl}$ and the Assumption \ref{ass_2rootrank} and \ref{ass_2rootstab} hold, the determinant is:
    \begin{equation}
    \begin{split}
    \det(A_\phi) 
    & = (\sgn(\tilde{\phi}_{il}) |\tilde{\phi}_{kl} \tilde{\phi}_{ij} \tilde{\phi}_{kj}| + \tilde{\phi}_{il}) \tilde{\phi}^2_{kl} \tilde{\phi}^2_{ij} \tilde{\phi}_{il} \\
    & = (|\tilde{\phi}_{il}| |\tilde{\phi}_{kl} \tilde{\phi}_{ij} \tilde{\phi}_{kj}| + \tilde{\phi}^2_{il}) \tilde{\phi}^2_{kl} \tilde{\phi}^2_{ij} >0.
    \end{split}
    \nonumber
    \end{equation}
    For $ A_x $ and $ A_y $, the analysis is analogous. Thus, all eigenvalues of $ D^h \hat{M}^h_{ff} $ are strictly positive.
\end{proof}
Theorem \ref{the_dsm} provides closed‑form stabilizing matrices for DEPs with $\ell \leq 2$. For longer DEPs, closed‑form solutions may not exist. However, by decomposing the 2‑rooted graph into multiple DEPs, the stabilizing matrix can be computed independently per DEP by solving \eqref{computing D}. This approach enables decentralized or parallel computation, thereby significantly accelerating the overall solution process.

\subsection{Complete Implementation Workflow} \label{sec_workflow}
Algorithm~\ref{alg:control-flow} summarizes the complete procedure for deploying the proposed distributed non‑uniform scaling formation maneuver control, from initial design to online execution.

\begin{algorithm}[t]
\caption{\HT{Complete Implementation Procedure for Distributed Non-uniform Scaling Formation Maneuver Control}}
\label{alg:control-flow}
\HT{
\begin{algorithmic}[1]
\STATE \textbf{Design Phase:}
\STATE Design a DEP-induced sensing graph $G$
\STATE If leaders are stationary:
\STATE \quad Design $(\tilde{g},\theta)$ satisfying Assumptions \ref{ass_2rootrank}, \ref{ass_2rootstab}
\STATE If leaders are moving:
\STATE \quad Design $(\tilde{g},\theta)$ satisfying Assumptions \ref{ass_2rootrank}, \ref{ass_c}

\STATE \textbf{Offline Preparation:}
\STATE Compute matrix‑valued Laplacian $M(G,\tilde{g},\theta)$ via (\ref{equ_mvl})
\STATE Compute stabilizing matrix $D = \diag(D_k)$ using the methods in Section \ref{sec_design_D}
\STATE Distribute to each agent $k$: $\{M_{ki}: i\in\mathcal{N}_k\}$ and $D_k$

\STATE \textbf{Online Execution (for each agent $k$):}
\STATE \textbf{while} not converged \textbf{do}
\STATE \quad Measure $\{p_i-p_k,\phi_i-\phi_k\}_{i\in\mathcal{N}_k}$
\STATE \quad \textbf{if} leaders are stationary \textbf{then}
\STATE \qquad Compute $u_k, \omega_k$ via (\ref{uk})
\STATE \quad \textbf{else} \ \textsl{(moving leaders)}
\STATE \qquad Measure $\{\dot{g}_i\}_{i\in\mathcal{N}_k}$
\STATE \qquad Compute $u_k, \omega_k$ via (\ref{equ_mluf})
\STATE \quad \textbf{end if}
\STATE \quad Update state via single-integrator dynamics \eqref{equ_dy}
\STATE \textbf{end while}
\end{algorithmic}}
\end{algorithm}

\begin{figure*}[t] 
	\centering  
	\includegraphics[width=7in]{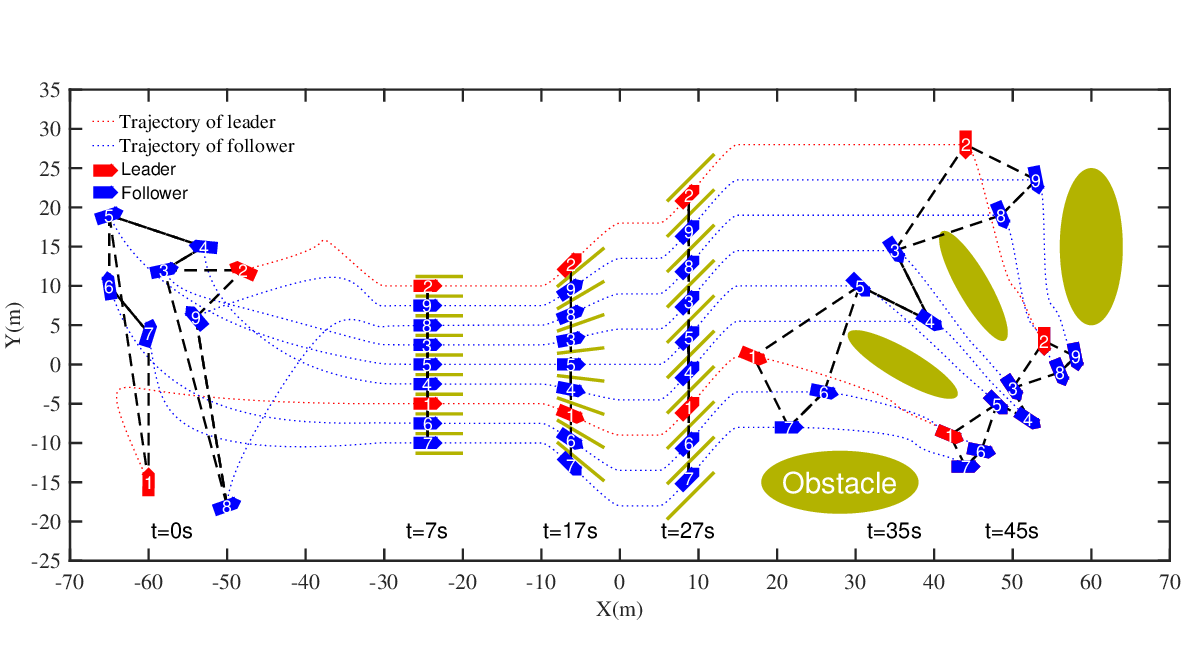}        
	\caption{Formation maneuver trajectories in 2-D space.}     
	\label{fig_trajg}     
\end{figure*}
\begin{figure}[!htbp]
	\centering	
	\includegraphics[width=3in]{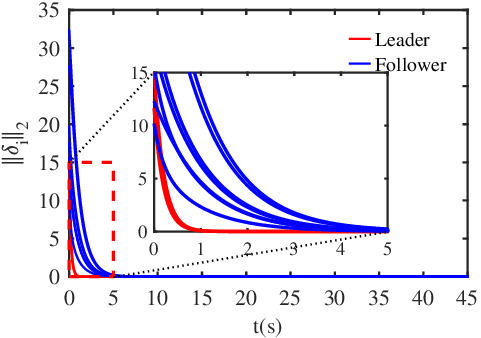}   	
	\caption{Tracking errors.}     
	\label{fig_te}      
\end{figure}


\section{A Simulation Example}\label{Simulation}
This section gives simulations to illustrate our results. We consider a nominal formation lying in $\mathbb{R}^2$ with DEP-induced graph shown in Fig.~\ref{fig_mdep}. The formation consists of two leaders $g_l= [g^{\top}_1, g^{\top}_2]^{\top}$ and seven followers $g_f= [g^{\top}_3, g^{\top}_4, g^{\top}_5, g^{\top}_6, g^{\top}_7, g^{\top}_8, g^{\top}_9]^{\top}$. \HT{The nominal configuration $(\tilde{g}, \theta)$ with $\theta = 0$ is given as: 
\[
\tilde{g}_1 = [-4, -2, \frac{ \pi}{8}]^{\top}, \tilde{g}_2 = [2, 4, -\frac{\pi}{4}]^{\top}, \tilde{g}_3 = [0, 1, -\frac{\pi}{16}]^{\top},
\]
\[
\tilde{g}_4 = [1, -1, \frac{\pi}{16}]^{\top}, \tilde{g}_5 = [-1, 0, 0]^{\top}, \tilde{g}_6 = [-2, -3, \frac{3\pi}{16}]^{\top},
\]
\[
\tilde{g}_7 = [-3, -4, \frac{\pi}{4}]^{\top}, \tilde{g}_8 = [3, 2, -\frac{\pi}{8}]^{\top}, \tilde{g}_9 = [4, 3, -\frac{3\pi}{16}]^{\top}.
\]}

The matrix-valued Laplacian $M(G,\tilde{g}, \theta)$ corresponding to the nominal formation $(G, \tilde{g}, \theta)$ can be calculated by (\ref{equ_mvl}), and the diagonal stabilizing matrix $D$ can be obtained based on Theorem \ref{the_dsm}. It is clear that this nominal formation satisfies Assumptions \ref{ass_2rootrank}, \ref{ass_2rootstab}, and \ref{ass_c}.

This simulation aims to validate a proposed control strategy for coordinated formation control of multiple agents navigating dense obstacles. The control goal is to enable leaders to track the predefined reference trajectory, defined by maneuver parameters in Table \ref{tab:my_label} with cubic spline interpolation for continuously differentiable trajectories, while followers maintain a desired geometric formation using controller \eqref{equ_mluf}.

The simulation results, depicted in Fig.~\ref{fig_trajg}, illustrate the dynamic evolution of the formation. The initial positions and yaw angles of the agents are randomly assigned. Upon activation, the multi-agent system achieves the first target formation within 5 seconds. At this stage, the line formation ($s_x = 0$) and attitude alignment ($s_\phi = 0$) are established. During 5-10 seconds, the position formation executes pure translation. Subsequently (10-15s), $s_\phi$ transitions from 0 to -1, inducing an attitude scaling transformation. During 15–20s, the team navigates the obstacles by scaling the position formation ($s_y:2.5 \rightarrow 4.5$) while maintaining the pre-configured attitudes from the previous phase. This scale-based avoidance strategy results in a tightly coordinated interplay between attitude and position formations, a capability unattainable by either technique in isolation, allowing the formation to navigate through a trumpet-shaped obstacle.

\begin{table}[t]
    \centering
    \begin{tabular}{ccccccc}
        \hline
        t(s)  & $s_{x}$ & $\tau_{x}$ & $s_{y}$ & $\tau_{y}$ & $s_{\phi}$ & $\tau_{\phi}$ \\[6pt]
        \hline
        0  & 5       & -47        & 2.5     & 4          & 1          & 0 \\
        5  & 0       & -30        & 2.5     & 0          & 0          & 0 \\
        10 & 0       & -15        & 2.5     & 0          & 0          & 0 \\
        15 & 0       & -10        & 2.5     & 0          & -1         & 0 \\
        20 & 0       & 0          & 4.5     & 0          & -1         & 0 \\
        25 & 0       & 5          & 4.5     & 0          & 0          & $\frac{\pi}{4}$ \\
        30 & 0       & 15         & 4.5     & 10         & 0          & $\frac{\pi}{4}$ \\
        35 & 4.5     & 35         & 4.5     & 10         & 1          & $-\frac{\pi}{4}$ \\
        45 & 2       & 50         & 2       & -5         & 1          & $-\frac{\pi}{4}$ \\
        \hline
    \end{tabular}
    \caption{Key maneuver parameters}
    \label{tab:my_label}
\end{table}


The formation then undergoes the following maneuvers.
\begin{itemize}
    \item 20--25s: The attitude formation realigns and executes a $\frac{\pi}{4}$ translational shift.
    \item 25--30s: The position formation performs simultaneous translations in both $x$ and $y$ directions.
    \item 30--35s: The formation performs a non-uniform scaling transformation, resulting in an enlarged formation pattern.
    \item 35--45s: The position formation translates and uniformly contracts, while maintaining fixed yaw angles.
\end{itemize}

By adjusting only leaders' positions and attitudes, the proposed control strategy enables continuous translations and non-uniform scalings of the joint position-attitude formation. Notably, by accounting for the physical shape of the agents (rather than modeling them as point masses), the proposed control method allows the formation to navigate through narrow arrays of parallel or non-parallel obstacles, as illustrated in the figure. In contrast, most existing approaches, e.g., \cite{Fang20241,Fang2024,Vu2024,Cheng2025,Huang2025}, require the team to make a detour, resulting in reduced efficiency and a lower likelihood of finding feasible paths in dense obstacle environments.

As evidenced by Fig.~\ref{fig_te}, the tracking errors converge asymptotically to zero, which validates the proposed control strategy and aligns with Theorem \ref{the_dd}.



\HT{
\section{Discussions for Practical Scenarios}
\label{sec:discussion}

\subsection{Collision Avoidance and Physical Size}\label{subsec:collision_avoidance}

The proposed control laws ensure global exponential convergence under sparse sensing, but do not explicitly handle collision avoidance. We address this separately for the formation shape control and formation maneuver control phases. 

For formation shaping from arbitrary initial positions, existing methods (e.g., artificial potential fields or control barrier functions) often require weakening the convergence guarantees or increasing the required sensing density. In practice, we adopt a two‑fold strategy: (i) for scenarios with sufficiently small initial errors, the proposed controller can directly ensure collision‑free convergence \cite{Zhao2016}; (ii) for challenging initial configurations, we first employ formation‑reshaping techniques \cite{ZhangFan2025,Kratky2025} to guide agents into a safe, near‑formation state, and then switch to the proposed controller.

During the formation maneuvering phase, collision-free motion can be ensured by a two-layer strategy. At the planning layer, the maneuver parameters \(s(t)\) and \(\tau(t)\) are designed such that the resulting target trajectory \(g^*(t)\) satisfies all inter‑agent safety‑distance constraints. To formulate these constraints, consider any pair of agents \(i\) and \(j\) with a safety radius \(r > 0\) (accounting for physical size). The squared distance between their target positions is:
\begin{equation}
\| p^*_i(t) - p^*_j(t) \|^2 = (s_x(t))^2 ( \tilde{p}_{ij,\theta}^x)^2 + (s_y(t))^2 ( \tilde{p}_{ij,\theta}^y)^2.
\nonumber
\end{equation}
Therefore, collision avoidance requires:
\begin{equation}
\label{eq:collision_constraint}
(s_x(t))^2 ( \tilde{p}_{ij,\theta}^x)^2 + (s_y(t))^2 ( \tilde{p}_{ij,\theta}^y)^2 \ge 4r^2, \quad \forall t,i \ne j.
\end{equation}
These inequalities define a feasible region in the \(s_x\)-\(s_y\) plane for each agent pair. Planning \(s_p(t)\) to stay inside the intersection of all such regions yields a collision‑free trajectory \(g^*(t)\). At the control layer, the proposed distributed laws guarantee exponential tracking of \(g^*(t)\), and if the tracking errors are sufficiently small, actual collisions are avoided.



\subsection{Impact of Practical Factors on the Sensing Graph}
\label{subsec:practical_graph_issues}
The theoretical results rely on a fixed, bidirectional, and 2-rooted sensing graph $G$. Practical deployments must consider factors that may affect these assumptions.

\subsubsection{Sensing Range and Directional Constraints}
\label{subsubsec:sensing_range_dir}
To maintain sensing within a maximum range \(r_{\max}\), \(s_p(t)\) must satisfy
\begin{equation}
\label{eq:range_constraint}
(s_x(t))^2(\tilde p_{ij,\theta}^x)^2+(s_y(t))^2(\tilde p_{ij,\theta}^y)^2 \le r_{\max}^2,\forall t,(i,j)\in E.
\end{equation}
While \(s_x(t) \neq s_y(t)\), the relative position vector
\begin{equation}
p_i(t) - p_j(t) = R(\theta) \diag(s_p(t)) R^\top(\theta) (\tilde{p}_i - \tilde{p}_j)
\nonumber
\end{equation}
generally changes direction unless aligned with the axes defined by \(\theta\). Fixed sensors may lose sight of neighbors during such maneuvers. Since each follower senses only two neighbors, equipping it with two steerable sensors—each actively tracking one neighbor—preserves required sensing links.

\subsubsection{Obstacle Occlusion}
\label{subsubsec:occlusion}
Occlusion breaks line‑of‑sight (LOS), threatening link integrity. Our simulations (Fig.~\ref{fig_trajg}) assume either low‑profile obstacles or the use of non‑line‑of‑sight (NLOS) technologies (e.g., millimeter-wave radar \cite{Zhu2023} or UWB \cite{Silva2020, Chen2025}) to maintain sensing links. Otherwise, online graph reconfiguration would be required—a direction for future work.


\subsubsection{Bidirectional Sensing Assumption}
\label{subsubsec:bidirectional_assumption}
The bidirectional sensing assumption is practically achievable through a combination of steerable sensors and NLOS technologies (Sections~\ref{subsubsec:sensing_range_dir} and \ref{subsubsec:occlusion}). Strictly, mutual sensing is required only between consecutive inner agents in DEPs of length $\ell \geq 2$. More critically, our distributed control relies solely on locally obtainable information, and it is theoretically compatible with systems where similar information is acquired via communication and cooperative estimation---as demonstrated in distributed swarm systems \cite{Zhou2022}. Such communication‑based implementations inherently ensure bidirectional information flow. Generalizing the theoretical framework to purely directed physical sensing graphs is a direction for future research.

\subsection{Extension to Complex Dynamics}\label{subsec:complex_dynamic}

The proposed approach is not applicable to systems with inherent kinematic coupling between position and attitude (e.g., unicycles or quadrotors) for simultaneous and independent formation control of both quantities.

Nevertheless, it is well-suited for systems with physically decoupled position and attitude actuation. This includes fully actuated platforms (e.g., omnidirectional robots), which can be modeled using single-integrator dynamics. More broadly, it enables control of composite platforms (e.g., fire‑fighting robots) that consist of a mobile base (which may be nonholonomic, e.g., differential‑drive, or possess complex nonlinear dynamics, e.g., quadrotor) and an independently steerable payload (e.g., a water cannon). In such systems, the payload orientation is coordinated via our attitude formation control (driven by single-integrator inputs, with active compensation for any base-induced disturbances), while standard techniques (e.g., in \cite{Fathian2021}) convert the single-integrator position commands into feasible base motion.

\subsection{Velocity and Acceleration Limits}
\label{subsec:dynamic_constraints}

To respect agent velocity and acceleration limits, the maneuver parameters must satisfy
\begin{equation}\label{eq:vel_constraint}
\| \dot{g}_i^*(t) \| = \left\| \dot{\tau}(t) + \varTheta \diag(\dot{s}(t)) \varTheta^{\top} \tilde{g}_i \right\| \leq v_{\max}, \, \forall i, t,
\end{equation}
\begin{equation}\label{eq:acc_constraint}
\| \ddot{g}_i^*(t) \| = \left\| \ddot{\tau}(t) + \varTheta \diag(\ddot{s}(t)) \varTheta^{\top} \tilde{g}_i \right\| \leq a_{\max}, \, \forall i, t,
\end{equation}
where $v_{\max}$ and $a_{\max}$ are the maximum speed and acceleration. Because these constraints depend only on the known nominal configuration $(\tilde{g},\theta)$ and the maneuver parameters, sufficiently smooth trajectories for $s(t)$ and $\tau(t)$ can be planned to meet the physical limits of the agents.



\subsection{Navigation in Complex Environments}
\label{subsec:nav_complex_env}

Navigation in complex environments can be formulated as a spatio-temporal trajectory optimization problem. Given the nominal formation \((G,\tilde g,\theta)\), the goal is to compute smooth time-varying maneuver parameters \(s(t)\) and \(\tau(t)\) that satisfy all safety, sensing, and dynamic constraints while minimizing a suitable cost.

In addition to the constraints introduced earlier, static obstacle avoidance is required. For each agent \(i\) (occupying a region \(\mathcal{A}_i(t)\) determined by its target state \(g^*_i(t)\)) and each obstacle region \(\mathcal{O}_k\), we impose
\begin{equation}
\label{eq:obstacle_constraint}
\operatorname{dist}\bigl(\mathcal{A}_i(t),\, \mathcal{O}_k\bigr) \ge r_a, \quad \forall i,k,t,
\end{equation}
where \(r_a > 0\) is a safety margin. The function \(\operatorname{dist}(\cdot,\cdot)\) denotes the minimum Euclidean distance between two sets.

The resulting optimization over horizon $[0,T]$ is
\begin{subequations}
\label{eq:full_planning_problem}
\begin{align}
\min_{s,\tau,T} \quad & \int_0^T \mathcal{L}(s,\tau) dt + \rho T \\
\text{s.t.} \quad & s(0)=s_0, \tau(0)=\tau_0, s(T)=s_T, \tau(T)=\tau_T, \\
& \dot s(0)=\dot \tau(0)=\dot s(T)=\dot \tau(T)=0, \\
& \eqref{eq:collision_constraint}, \eqref{eq:range_constraint}, \eqref{eq:vel_constraint}, \eqref{eq:acc_constraint}, \eqref{eq:obstacle_constraint}, \quad \forall t \in [0,T],
\end{align}
\end{subequations}
where $\mathcal{L}(s,\tau) = \|\dot s\|^2 + \|\dot \tau\|^2 + \|\ddot s\|^2 + \|\ddot \tau\|^2$. The boundary conditions \(s_0,\tau_0\) and \(s_T,\tau_T\) specify the initial and desired final formation shapes. The objective promotes smooth variations of the maneuver parameters while penalizing the total duration via the weight \(\rho \ge 0\). This constrained optimization can be solved numerically using established methods \cite{Hossein2022,Quan2023,Liu2024}. Consequently, the proposed distributed control law serves as the low-level tracking controller within a hierarchical autonomy stack, where a higher-level planner computes feasible \(s(t)\) and \(\tau(t)\) by solving \eqref{eq:full_planning_problem}. Investigating the tight integration between such planners and our matrix-valued-Laplacian-based control law is a promising direction for future work.}

\section{Conclusion}\label{Conclusion}
We have proposed a novel distributed leader-follower formation maneuver control approach for multi-agent systems in the plane, enabling simultaneous non-uniform scaling and translational maneuvers of a joint position-attitude formation. A matrix-valued Laplacian has been developed to characterize the target configuration space, and the nominal formation was shown to achieve maximum maneuverability if and only if the underlying sensing graph is 2-rooted. Additionally, 
by decomposing the graph into multiple DEPs, a scalable approach for the stabilizing matrix design was proposed. 
Simulation results have validated the effectiveness of the control strategy, showing that closed-loop errors converge globally to zero and adaptive formation maneuvers are achieved in dense obstacle scenarios. Future work will focus on designing controllers that leverage more sophisticated attitude transformations and enhance resilience to agent or edge failures, all without relying on a global coordinate system, thereby bridging the gap between theoretical advancements and practical deployment in real-world multi-agent systems.


\section{Appendix}
\subsection{Proof of Lemma \ref{lem:mp2rg_equivalence}} \label{sec:proofLemmp2rg_equivalence}
\begin{apxproof}
	\textit{(Sufficiency)} Suppose $G$ contains a spanning DEP-induced graph $\mathcal{L}_{\kappa}$ constructed recursively as in Definition~\ref{def:multi_dual_entry_graph}. By definition, any agent in $\mathcal{G}_{\mathcal{P}_h}$ has two disjoint bidirectional paths in $\mathcal{L}_h$ to agents 1 and 2. Since $\mathcal{L}_{\kappa}$ spans $G$, every agent in $G$ is 2-reachable from $\{1,2\}$, i.e., $G$ is 2-rooted.

	
	\textit{(Necessity)} Assume $G$ is 2-rooted with roots $\{1, 2\}$. Initialize $\mathcal{L}_0 = (V_0, E_0)$, where $V_0 = \{1, 2\}$ and $E_0 = \emptyset$. For any agent $k \notin V_0$, since $k$ is 2-reachable from $\{1, 2\}$, there must exist two disjoint paths from $1$ and $2$ to $k$, the union of these paths with involved vertices must contain a DEP $G_{\mathcal{P}_1} = (V_{\mathcal{P}_1}, E_{\mathcal{P}_1})$ with entry agents $i_1=1$, $j_1=2$ and $\ell_1 \ge 1$ inner agents labeled $\{|V_0| + 1, \dots, |V_0| + \ell_1\}$. Construct $\mathcal{L}_1 = (V_1, E_1)$ with $V_1 = V_0 \cup V_{\mathcal{P}_1}$ and $E_1 = E_0 \cup E_{\mathcal{P}_1}$. 
	
	Next, select an agent $l \notin V_1$ that is 2-reachable from $\{1, 2\}$. There exist two disjoint paths from distinct agents $i_2, j_2 \in V_1$ to $l$, with all intermediate vertices distinct from $V_1$. The union of these paths with $l$ must contain a DEP $G_{\mathcal{P}_2} = (V_{\mathcal{P}_2}, E_{\mathcal{P}_2})$ with entry agents $\{i_2, j_2\}$ and $\ell_2 \ge 1$ inner agents labeled $\{|V_1| + 1, \dots, |V_1| + \ell_2\}$. Construct $\mathcal{L}_2 = (V_2, E_2)$ with $V_2 = V_1 \cup V_{\mathcal{P}_2}$ and $E_2 = E_1 \cup E_{\mathcal{P}_2}$.
	
	Repeat the above process until all agents in $V$ are included in some $\mathcal{L}_{\kappa}$. The resulting graph $\mathcal{L}_{\kappa}$ is a DEP-induced subgraph by Definition~\ref{def:multi_dual_entry_graph}.
\end{apxproof}

\subsection{Proof of Lemma \ref{lem:equivalence}} \label{sec:proofLemEquivalence}
    By Definition~\ref{defnon-singularConfiguration}, condition (a) holds if the mapping $ [s^\top, \tau^\top]^\top \mapsto A [s^\top, \tau^\top]^\top = g $ is injective, i.e., $ \nullspace(A) = \{0\} $. For $ A \in \mathbb{R}^{3n \times 6} $, the rank-nullity theorem implies $ \nullspace(A) = \{0\} \iff \rank(A) = 6 $ (condition (b)). Since $ \varPi(\tilde{g}, \theta) = \image(A) $, we have $ \rank(A) = 6 \iff \dim(\varPi(\tilde{g}, \theta)) = 6 $ (condition (c)). Thus, conditions (a), (b), and (c) are equivalent. When $\rank(A) < 6$, a singular configuration leads to ineffective parameters. Next, we prove the equivalence between condition (b) and condition (d) by establishing both implications.

	Since $I_n \otimes \varTheta$ is non-singular, we have
    \begin{equation}
        \rank(A) = \rank((I_n \otimes \varTheta) \cdot \bar{A}) = \rank(\bar{A}),
        \nonumber
    \end{equation}
	where $\bar{A} = 
	\begin{bmatrix}
	\diag(\tilde{g}_{1,\theta}) & \varTheta^\top \\
	\diag(\tilde{g}_{2,\theta}) & \varTheta^\top \\
	\vdots & \vdots \\
	\diag(\tilde{g}_{n,\theta}) & \varTheta^\top
	\end{bmatrix} \in \mathbb{R}^{3n \times 6}.$

\begin{proof}
[$(b) \Rightarrow (d)$]
    Suppose $\{\tilde{p}^x_{i,\theta}\}_{i \in V}$ is a singleton, i.e., $\tilde{p}_{i,\theta}^x = c$ for some constant $c$ and all $i$. Then the first column of $\bar{A}$ is a constant vector:
        \begin{equation}
            v_1 = [c,\ 0,\ 0,\ c,\ 0,\ 0,\ \dots,\ c,\ 0,\ 0]^\top.
            \nonumber
        \end{equation}
        This vector can be written as a linear combination of the 4th and 5th columns of $\bar{A}$, denoted as $v_4$ and $v_5$, respectively:
        \begin{equation}
            v_1 = -c\cos\theta\,v_4 + c\sin\theta\,v_5.
            \nonumber
        \end{equation}
        Hence, $v_1$ is linearly dependent on other columns, implying $\rank(\bar{A}) = \rank(A) < 6$. The same argument applies if $\{\tilde{p}_{i,\theta}^y\}_{i \in V}$ or $\{\tilde{\phi}_i\}_{i \in V}$ is a singleton. Therefore, all three sets must not be singletons.

$(d) \Rightarrow (b)$:
    Assume each of the sets $\{\tilde{p}^x_{i,\theta}\}_{i \in V}$, $\{\tilde{p}^y_{i,\theta}\}_{i \in V}$, and $\{\tilde{\phi}_i\}_{i \in V}$ is not a singleton. This implies the existence of distinct indices $i_k,j_k$ ($k=1,2,3$) such that:
    \begin{align}
        \tilde{p}^x_{i_1,\theta} \neq \tilde{p}^x_{j_1,\theta}, \quad 
        \tilde{p}^y_{i_2,\theta} \neq \tilde{p}^y_{j_2,\theta}, \quad 
        \tilde{\phi}_{i_3} \neq \tilde{\phi}_{j_3}.
        \nonumber
    \end{align}

    Next, apply rank-preserving operations to matrix $\bar{A}$: subtract row $ 3i_1-2 $ from $ 3j_1-2 $, row $ 3i_2-1 $ from $ 3j_2-1 $, and row $ 3i_3 $ from $ 3j_3 $. Consider the $6 \times 6$ submatrix with rows $ 3j_1-2 $, $ 3j_2-1 $, $ 3j_3 $, $ 3i_1-2 $, $ 3i_2-1$, $ 3i_3 $:
    \begin{equation}
        \begin{bmatrix}
            \diag([\tilde{p}^x_{j_1,\theta} - \tilde{p}^x_{i_1,\theta}, \tilde{p}^y_{j_2,\theta} - \tilde{p}^y_{i_2,\theta}, \tilde{\phi}_{j_3} - \tilde{\phi}_{i_3}]^{\top}) & 0 \\
            \diag([\tilde{p}^x_{i_1,\theta}, \tilde{p}^y_{i_2,\theta}, \tilde{\phi}_{i_3}]^{\top}) & \varTheta^\top \\
        \end{bmatrix}.
        \nonumber
    \end{equation}
    Since $\tilde{p}^x_{j_1,\theta} - \tilde{p}^x_{i_1,\theta}, \tilde{p}^y_{j_2,\theta} - \tilde{p}^y_{i_2,\theta}, \tilde{\phi}_{j_3} - \tilde{\phi}_{i_3} \neq 0$, and $\varTheta^\top$ is invertible, we have $\rank(\bar{A}) = \rank(A) = 6$.
\end{proof}

\subsection{Proof of Theorem \ref{the:rankL}} \label{sec:proofTheRankL}
We begin by introducing two lemmas.
\begin{lemma}[Diagonal Stability \cite{Anderson2009}, Theorem 3.2] \label{lem:DiagonalStability}
    Let $A$ be an $n \times n$ matrix whose all leading principal minors are nonzero. 
    Then, there exists a diagonal matrix $D$ such that every eigenvalue of $DA$ has a positive real part.
\end{lemma}

\begin{lemma}\label{lem:PathG} 
Consider a nominal formation $(G_{\mathcal{P}_h},\tilde{g}, \theta)$ in $\mathbb{R}^2$, where $G_{\mathcal{P}_h} = (V_{\mathcal{P}_h}, E_{\mathcal{P}_h})$ is a DEP graph with entry agents $\{i_h, j_h\}$ and inner agents locally labeled $\{1, \dots, \ell_h\}$ as defined in Definition~\ref{def:dual_entry_path}. Then:

\begin{enumerate}
    \item \label{lem:PathG-non-singular}
    The matrix $\hat{M}_{ff}$ is non-singular if and only if the condition in \eqref{eq nondegeneracy condition} is satisfied.

    \item \label{lem:PathG-stability}
    Under conditions \eqref{eq nondegeneracy condition} and \eqref{eq stabilitycondition}, there exists a diagonal matrix $D$ such that every eigenvalue of $D\hat{M}_{ff}$ has a positive real part.
\end{enumerate}
\end{lemma}

\begin{proof}
    Let $g_{\theta}=\diag(\varTheta^{\top}) g=[\cdots,g_{i,\theta}^{\top},\cdots]^{\top}$, where $g_{i,\theta}=\varTheta^{\top} g_i=[p^x_{i,\theta}, p^y_{i,\theta},\phi_i]^{\top}$. Define stacked vectors $p^x_{\theta}=[\cdots,p^x_{i,\theta},\cdots]^{\top}$ and $p^y_{\theta}=[\cdots,p^y_{i,\theta},\cdots]^{\top}$. Since each constant value matrix block $w_{ij}$ defined in \eqref{equ_vvc} is a diagonal matrix, there exist a row permutation matrix $Q$ and a column permutation matrix $P$ such that 
    \begin{equation} \label{equ_de}
    \begin{split}
        &Q \hat{M}_{f} P= Q \hat{M}_{f} \begin{bmatrix}
        P_{ll} & 0  \\
        0 & P_{ff} 
        \end{bmatrix} \\
        &=\begin{bmatrix}
        \hat{M}_{fl}^x & 0  & 0 & \hat{M}_{ff}^x & 0 & 0 \\
        0 & \hat{M}_{fl}^y & 0 & 0  & \hat{M}_{ff}^y & 0\\
        0 & 0 & \hat{M}_{fl}^\phi & 0 & 0 & \hat{M}_{ff}^\phi
        \end{bmatrix}.
    \end{split}
    \end{equation}
    In other words, $\hat{M}_f$ can be decomposed into independent constraint matrices for each state component, we have
    \begin{equation}
    \begin{cases}
    &\hat{M}_f^x p^x_{\theta} = 0, \\
    &\hat{M}_f^y p^y_{\theta} = 0, \\
    &\hat{M}_f^{\phi} \phi = 0,
    \end{cases}
    \nonumber
    \end{equation}
    where the matrices $\hat{M}_f^x = [\hat{M}_{fl}^x, \hat{M}_{ff}^x]$, $\hat{M}_f^y = [\hat{M}_{fl}^y, \hat{M}_{ff}^y]$, and $\hat{M}_f^\phi = [\hat{M}_{fl}^\phi, \hat{M}_{ff}^\phi]$ are partitioned according to the leader-follower structure. From \eqref{equ_de}, it holds that
    
    \begin{equation}\label{eq mff}
       Q\hat{M}_{ff} P_{ff} = \begin{bmatrix}
        \hat{M}_{ff}^x & 0 & 0 \\
        0  & \hat{M}_{ff}^y & 0\\
        0 & 0 & \hat{M}_{ff}^\phi
        \end{bmatrix}. 
    \end{equation}

    \textit{Proof of Part \ref{lem:PathG-non-singular})}:
    From \eqref{eq mff}, we have
    \begin{equation}
        \rank(\hat{M}_{ff}) =\rank(\hat{M}^x_{ff}) + \rank(\hat{M}^y_{ff}) + \rank(\hat{M}^{\phi}_{ff}).
        \nonumber
    \end{equation}
    
    Consequently, $\hat{M}_{ff}$ is non-singular if and only if $\hat{M}^x_{ff}$, $\hat{M}^y_{ff}$, and $\hat{M}^{\phi}_{ff}$ are all non-singular. Next, we establish the conditions under which $\hat{M}^{\phi}_{ff}$ is non-singular.

    To simplify the notation, we adopt simplified indices by mapping the original agent labels $i_h,j_h,1,2,\cdots, \ell_h$ to consecutive integers $1, 2, 3, \cdots, n$. Under this notation, the matrix $\hat{M}_f^{\phi}$ takes the following form:
    \begin{equation}
        \begin{aligned}
        &\hat{M}_f^{\phi}=\left[ \begin{array}{c|c}
        \hat{M}^{\phi}_{fl} & \hat{M}^{\phi}_{ff}
        \end{array} \right] = \\
        & \left[ \begin{array}{cc|ccccc}
        \phi_{43} & 0            & \phi_{14} & \phi_{31} & 0         & \cdots    & 0      \\
        0         & 0            & \phi_{54} & \phi_{35} & \phi_{43} & \ddots    & \vdots \\
        0         & 0            & 0         & \phi_{65} & \phi_{46} & \ddots    & 0      \\
        \vdots    & \vdots       & \vdots    & \ddots    & \ddots    & \ddots    & \phi_{n_1n_2} \\
        0         & \phi_{nn_1}  & 0         & \cdots    & 0         & \phi_{2n} & \phi_{n_12} \rule{0pt}{0.5cm}
        \end{array} \right]
        \end{aligned},
        \nonumber
    \end{equation}
    where $\hat{M}^{\phi}_{fl} \in \mathbb{R}^{(n-2) \times 2}$, $\hat{M}^{\phi}_{ff} \in \mathbb{R}^{(n-2) \times (n-2)}$, $\phi_{ij} = \phi_i - \phi_j$ and $n_i$ is an abbreviation for $n-i$. 

    It is clear that $\hat{M}^{\phi}_{ff}=[m_{ij}]$ is a tridiagonal matrix. Let $f_0=1$, $f_1 = \det([\phi_{14}])=\phi_{14}$, and $f_{n-2} = \det(\hat{M}^{\phi}_{ff})$. According to \cite[Theorem 2.1]{Moawwad2004}, $\det(\hat{M}^{\phi}_{ff})$ can be computed from a three-term recurrence relation 
    \HT{
    \begin{equation}
        f_i =m_{ii} f_{i-1} - m_{i(i-1)}m_{(i-1)i} f_{i-2}, i=2,3,\cdots,n-2,
        \nonumber
    \end{equation}   
    where $f_i$ denotes leading principal minor of order $i$. Next, we prove this result by induction. For the base case \(i=2\):
    \begin{equation}
    \begin{split}
    f_2 &=\phi_{35} f_1 - \phi_{54}\phi_{31} f_0 \\
    &=(\phi_{34} + \phi_{45}) f_1 - \phi_{54}\phi_{31} \\
    &=\phi_{34}\phi_{14} + \phi_{45}(\phi_{14} + \phi_{31}) \\
    &=\phi_{34}(\phi_{14} + \phi_{45}) \\
    &=\phi_{34}\phi_{15}.
    \end{split}
    \nonumber
    \end{equation}

    Assume that 
    \[
    \begin{aligned}
    f_{n-4} &= \phi_{34}\phi_{45}\cdots\phi_{(n-3)(n-2)}\phi_{1(n-1)},\\
    f_{n-3} &= \phi_{34}\phi_{45}\cdots\phi_{(n-2)(n-1)}\phi_{1n}.
    \end{aligned}
    \]
    Then we have}
    \begin{equation}
    \begin{split}
    f_{n-2} &=\phi_{(n-1)2} f_{n-3} - \phi_{2n}\phi_{(n-1)(n-2)} f_{n-4} \\
    &=\phi_{(n-1)n}f_{n-3}+ \phi_{n2}(f_{n-3} +\phi_{(n-1)(n-2)} f_{n-4}) \\
    &=\phi_{(n-1)n}f_{n-3} + \phi_{n2}(\phi_{34}\phi_{45}\cdots\phi_{(n-2)(n-1)}\phi_{1n}\\
    &+\phi_{34}\phi_{45}\cdots\phi_{(n-3)(n-2)}\phi_{1(n-1)}\phi_{(n-1)(n-2)}) \\
    &=\phi_{(n-1)n}f_{n-3} + \phi_{34}\phi_{45}\cdots\phi_{(n-1)n}\phi_{n2} \\
    &=\phi_{34}\phi_{45}\cdots\phi_{(n-1)n}\phi_{12}.
    \end{split}
    \nonumber
    \end{equation}
    So, $\phi_{34}\phi_{45}\cdots\phi_{(n-1)n}\phi_{12} \ne 0 \Longleftrightarrow \det(\hat{M}^{\phi}_{ff}) \ne 0 \Longleftrightarrow \rank(\hat{M}^{\phi}_{ff}) = n-2 \Longleftrightarrow \hat{M}^{\phi}_{ff}$ is non-singular.
    \HT{
    Similar to the above proof, we conclude that $\hat{M}^x_{ff}$ and $\hat{M}^y_{ff}$ are non-singular if and only if 
    \begin{equation}
        \tilde{p}^x_{34,\theta} \tilde{p}^x_{45,\theta} \cdots \tilde{p}^x_{(n-1)n,\theta}\tilde{p}^x_{12,\theta} \ne 0
        \nonumber
    \end{equation}
    and
    \begin{equation}
        \tilde{p}^y_{34,\theta} \tilde{p}^y_{45,\theta} \cdots \tilde{p}^y_{(n-1)n,\theta}\tilde{p}^y_{12,\theta} \ne 0.
        \nonumber
    \end{equation}}

    \textit{Proof of Part \ref{lem:PathG-stability})}:
    Since $Q \hat{M}_{ff} P_{ff}$ is block-diagonal as shown in \eqref{eq mff}, we analyze the submatrices $\hat{M}_{ff}^x$, $\hat{M}_{ff}^y$, and $\hat{M}_{ff}^\phi$. By Lemma~\ref{lem:DiagonalStability}, for each submatrix (e.g., $\hat{M}_{ff}^x$), there exists a diagonal matrix $D^x$ such that every eigenvalue of $D^x \hat{M}_{ff}^x$ has a positive real part if its all leading principal minors are nonzero. The same applies to $\hat{M}_{ff}^y$ and $\hat{M}_{ff}^\phi$ with diagonal matrices $D^y$ and $D^\phi$, respectively. Construct $D = \diag(D^x, D^y, D^\phi)$, which is diagonal and ensures that every eigenvalue of $D Q \hat{M}_{ff} P_{ff} = \diag(D^x \hat{M}_{ff}^x, D^y \hat{M}_{ff}^y, D^\phi \hat{M}_{ff}^\phi)$ has a positive real part, since every eigenvalue of each block has a positive real part. Let $D' = DQ$, we note that since $Q$ is a permutation matrix and $D$ is diagonal, $D'$ remains diagonal. Furthermore, $D Q \hat{M}_{ff} P_{ff}$ and $D'\hat{M}_{ff}$ share identical eigenvalues because $P_{ff}$ is also a permutation matrix.
    
    Next, we establish the conditions under which the leading principal minors of $\hat{M}_{ff}^x$, $\hat{M}_{ff}^y$, and $\hat{M}_{ff}^\phi$ are nonzero.

    From the proof of Part 1), all leading principal minors of $\hat{M}^{\phi}_{ff}$ are distinct from zero $\Longleftrightarrow f_1 \ne 0 \wedge f_2 \ne 0 \wedge \cdots \wedge f_{n-2} \ne 0 \Longleftrightarrow \phi_{14} \ne 0 \wedge \phi_{34}\phi_{15} \ne 0 \wedge \cdots \wedge \phi_{34}\phi_{45}\cdots\phi_{(n-1)n}\phi_{12} \ne 0 \Longleftrightarrow \phi_{34}\phi_{45}\cdots\phi_{(n-1)n}\phi_{14}\phi_{15}\cdots \phi_{1(n-1)} \phi_{1n} \phi_{12} \ne 0$. These conditions guarantee that $\hat{M}_{ff}^\phi$ has full rank and its leading principal minors are nonzero. The corresponding conditions for $\hat{M}_{ff}^x$ and $\hat{M}_{ff}^y$ follow similarly.
\end{proof}

\begin{proof}[Proof of Theorem \ref{the:rankL}]
    \textit{(Sufficiency)} According to Definition \ref{def:multi_dual_entry_graph}, the matrix $\hat{M}_f$ of DEP-induced graph $\mathcal{L}_{\kappa}$ takes the following form:
    \begin{equation} \label{equ_mf}
    \hat{M}_f =[\hat{M}_{fl} \, \hat{M}_{ff}]=\left[  \begin{array}{c|ccccc}
    \hat{M}_{fl}^1  & \hat{M}_{ff}^1  & 0      & \cdots & 0      \\
    *      & *      & \hat{M}_{ff}^2  & \ddots & \vdots \\
    \vdots & \vdots & \ddots & \ddots & 0      \\
    *      & *      & \cdots & *      & \hat{M}_{ff}^{\kappa}
    \end{array} \right] ,
    \end{equation}
    where $\hat{M}_{fl} \in \mathbb{R}^{(3n-6)\times6}$, $\hat{M}_{ff} \in \mathbb{R}^{(3n-6)\times(3n-6)}$, and $\hat{M}_{ff}^h, h \in \{1,2,\dots, \kappa\}$ are the corresponding blocks of the DEP graph $G_{\mathcal{P}_h}$. If $G$ satisfies Assumption \ref{ass_2rootrank}, by applying Lemma \ref{lem:PathG}, we have $\rank(\hat{M}_{ff}^h) = 3|V_{\mathcal{P}_h}|-6$. Considering the particular structure of $\hat{M}_f$, we know that
    \begin{equation}
    \rank(\hat{M}_{ff}) = \sum\limits_{h=1}^{\kappa} \rank(\hat{M}_{ff}^h) = 3n-6.
    \nonumber
    \end{equation}
    Thus $\hat{M}_{ff}$ is square and full rank, hence non-singular.

    \textit{(Necessity)} Suppose $G$ is not 2-rooted, implying that the removal of a particular agent results in some agents becoming unreachable from the root subset. For the sake of argument, assume that upon removing agent $i$, there emerges a subset $U$ comprising $i-1$ agents that are disconnected from all roots, and a complementary set $\bar{U}$ consisting of $n-i$ agents that remain accessible from at least one root. We can reindex the agents in $U$ as $1, \ldots, i-1$ and those in $\bar{U}$ as $i+1, \ldots, n$. Then $\hat{M}^u_f$ associated with $U$ adopts the following structure:
    \begin{equation}
    \left[
    \begin{array}{ccccccc}
    \hat{M}_{uu} & \hat{M}_{ui} & 0 \\
    \end{array}
    \right],
    \nonumber
    \end{equation}
    where $\hat{M}_{uu} \in \mathbb{R}^{(3i-3) \times (3i-3)}$ and $\hat{M}_{ui} \in \mathbb{R}^{(3i-3) \times 3}$. Denote the relabeled $g$ by $[g_{\alpha}^{\top}, g_{\beta}^{\top}]^{\top}$ where $g_{\alpha} \in \mathbb{R}^{3i \times 1}$ and $g_{\beta} \in \mathbb{R}^{3(n-i) \times 1}$. By the definition of $\hat{M}_f$ and Lemma \ref{lem:2}, we have
    \begin{equation}
    [\hat{M}_{uu} \, \hat{M}_{uk}] \diag(\varTheta^{\top}) \left( (I_i \otimes S) g_{\alpha} + 1_i \otimes \tau \right) = 0.
    \nonumber
    \end{equation}
    This implies $\rank([\hat{M}_{uu} \, \hat{M}_{ui}]) < 3i-3$, hence $[\hat{M}_{uu} \, \hat{M}_{ui} \, 0]$ is not of full row rank. Consequently, $\hat{M}_f$ is not of full row rank, which entails that $\hat{M}_{ff}$ is singular. This contradicts the statement that $\hat{M}_{ff}$ is non-singular. Therefore, $G$ is 2-rooted. According to Lemma \ref{lem:mp2rg_equivalence}, $G$ contains a spanning DEP-induced graph $\mathcal{L}_{\kappa}$. By \eqref{equ_mf} and Lemma \ref{lem:PathG}, we conclude that if $\hat{M}_{ff}$ is non-singular, then the nominal formation $(G, \tilde{g}, \theta)$ must satisfy Assumption \ref{ass_2rootrank}. The proof is completed.
\end{proof}
\HT{
\subsection{Proof of Theorem \ref{th convergence}} \label{sec:proofTheconvergence}
\begin{apxproof}
    Substituting (\ref{equ_sluf}) into $\dot{\delta}_f$ gives
    \begin{equation}
    \begin{split}
    \dot{\delta}_f &= (\hat{M}_{ff} \diag(\varTheta^{\top}))^{-1} \hat{M}_{fl} \diag(\varTheta^{\top}) \dot{g}^*_l + \dot{g}_f \\
    &=-\diag(\varTheta) D \hat{M}_{ff} \diag(\varTheta^{\top}) \delta_f.
    \end{split}
    \nonumber
    \end{equation}
    
    We first establish that under Assumptions \ref{ass_2rootrank} and \ref{ass_2rootstab}, there exists a diagonal matrix $D$ such that every eigenvalue of $D\hat{M}_{ff}$ has a positive real part.

    From \eqref{equ_mf}, the spectrum of $\hat{M}_{ff}$ is determined by its block diagonal components $\hat{M}_{ff}^h$ (where $h \in {1,2,...,\kappa}$), each corresponding to the DEP graph $G_{\mathcal{P}_h}$. Under Assumptions \ref{ass_2rootrank} and \ref{ass_2rootstab}, Lemma \ref{lem:PathG} guarantees that for each $\hat{M}_{ff}^h$, there exists a diagonal $D^h$ such that $\sigma(D^h\hat{M}_{ff}^h)$ has eigenvalues with positive real parts, where $\sigma(\cdot)$ denotes the spectrum.
    
    Taking $D = \diag( D^h )$ yields that $D\hat{M}_{ff}$ has the spectrum $\bigcup_{k=1}^\kappa \sigma(D^h\hat{M}_{ff}^h)$. Since all eigenvalues within each block have positive real parts, and blocks correspond to different path graphs, the combined spectrum maintains these properties.

    Next, since $\diag(\varTheta)$ is non-singular, the matrices $D \hat{M}_{ff}$ and $\diag(\varTheta) D \hat{M}_{ff} \diag(\varTheta^{\top})$ share the same eigenvalues. Consequently, all eigenvalues of $-\diag(\varTheta) D \hat{M}_{ff} \diag(\varTheta^{\top})$ lie in the open left half plane. This implies that the tracking error $\delta_f$ converges to zero globally and exponentially. 
\end{apxproof}
}

\bibliographystyle{IEEEtran}
\bibliography{references} 

@article{Fang20241,
   author = {Fang, X. and Xie, L. and Li, X.},
   title = {Integrated Relative-Measurement-Based Network Localization and Formation Maneuver Control},
   journal = {IEEE Transactions on Automatic Control},
   volume = {69},
   number = {3},
   pages = {1906-1913},
   DOI = {10.1109/TAC.2023.3330801},
   year = {2024}
}

@article{Moawwad2004,
   author = {El-Mikkawy, Moawwad E. A.},
   title = {On the inverse of a general tridiagonal matrix},
   journal = {Applied Mathematics and Computation},
   volume = {150},
   number = {3},
   pages = {669-679},
   DOI = {10.1016/S0096-3003(03)00298-4},
   year = {2004}
}

@article{Lin2014,
   author = {Lin, Z. and Wang, L. and Han, Z. and Fu, M.},
   title = {Distributed Formation Control of Multi-Agent Systems Using Complex Laplacian},
   journal = {IEEE Transactions on Automatic Control},
   volume = {59},
   number = {7},
   pages = {1765-1777},
   DOI = {10.1109/TAC.2014.2309031},
   year = {2014}
}

@ARTICLE{Fang2022,
  author={Fang, Xu and Li, Xiaolei and Xie, Lihua},
  journal={IEEE Transactions on Cybernetics}, 
  title={Distributed Formation Maneuver Control of Multiagent Systems Over Directed Graphs}, 
  year={2022},
  volume={52},
  number={8},
  pages={8201-8212},
  doi={10.1109/TCYB.2020.3044581}
}

@ARTICLE{Chen2023,
  author={Chen, Liangming and Cao, Ming},
  journal={IEEE Transactions on Automatic Control}, 
  title={Angle Rigidity for Multiagent Formations in 3-{D}}, 
  year={2023},
  volume={68},
  number={10},
  pages={6130-6145},
  doi={10.1109/TAC.2023.3237799}
}

@ARTICLE{Trinh2020,
  author={Trinh, Minh Hoang and Van Tran, Quoc and Ahn, Hyo-Sung},
  journal={IEEE Transactions on Automatic Control}, 
  title={Minimal and Redundant Bearing Rigidity: Conditions and Applications}, 
  year={2020},
  volume={65},
  number={10},
  pages={4186-4200},
  doi={10.1109/TAC.2019.2958563}
}

@article{Jing2019,
   author = {Jing, Gangshan and Zhang, Guofeng and Lee, Heung Wing Joseph and Wang, Long},
   title = {Angle-based shape determination theory of planar graphs with application to formation stabilization},
   journal = {Automatica},
   volume = {105},
   pages = {117-129},
   DOI = {10.1016/j.automatica.2019.03.026},
   year = {2019}
}

@article{Zhao2016,
   author = {Zhao, Shiyu and Zelazo, Daniel},
   title = {Bearing Rigidity and Almost Global Bearing-Only Formation Stabilization},
   journal = {IEEE Transactions on Automatic Control},
   volume = {61},
   number = {5},
   pages = {1255-1268},
   DOI = {10.1109/tac.2015.2459191},
   year = {2016}
}

@article{Buckley2021,
   author = {Buckley, Ian and Egerstedt, Magnus},
   title = {Infinitesimal Shape-Similarity for Characterization and Control of Bearing-Only Multirobot Formations},
   journal = {IEEE Transactions on Robotics},
   volume = {37},
   number = {6},
   pages = {1921-1935},
   DOI = {10.1109/tro.2021.3072549},
   year = {2021}
}

@article{Cao2020,
   author = {Cao, Kun and Han, Zhimin and Li, Xiuxian and Xie, Lihua},
   title = {Ratio-of-distance rigidity theory with application to similar formation control},
   journal = {IEEE Transactions on Automatic Control},
   volume = {65},
   number = {6},
   pages = {2598--2611},
   DOI = {10.1109/TAC.2019.2938318},
   year = {2020}
}

@article{jing2018,
   author = {Jing, Gangshan and Zhang, Guofeng and Lee, Heung Wing Joseph and Wang, Long},
   title = {Weak rigidity theory and its application to formation stabilization},
   journal = {SIAM Journal on Control and Optimization},
   volume = {56},
   number = {3},
   pages = {2248--2273},
   DOI = {10.1137/17M1122049},
   year = {2018}
}

@article{Lin2016,
   author = {Lin, Zhiyun and Wang, Lili and Chen, Zhiyong and Fu, Minyue and Han, Zhimin},
   title = {Necessary and sufficient graphical conditions for affine formation control},
   journal = {IEEE Transactions on Automatic Control},
   volume = {61},
   number = {10},
   pages = {2877--2891},
   DOI = {10.1109/TAC.2015.2504265},
   year = {2016}
}

@article{Morbidi2022,
   author = {Morbidi, Fabio},
   title = {Functions of the Laplacian Matrix With Application to Distributed Formation Control},
   journal = {IEEE Transactions on Control of Network Systems},
   volume = {9},
   number = {3},
   pages = {1459--1467},
   DOI = {10.1109/TCNS.2021.3113263},
   year = {2022}
}

@ARTICLE{Zhao2018,
  author={Zhao, Shiyu},
  journal={IEEE Transactions on Automatic Control}, 
  title={Affine Formation Maneuver Control of Multiagent Systems}, 
  year={2018},
  volume={63},
  number={12},
  pages={4140-4155},
  doi={10.1109/TAC.2018.2798805}
}

@article{yang2019,
   author = {Yang, Qingkai and Sun, Zhiyong and Cao, Ming and Fang, Hao and Chen, Jie},
   title = {Stress-matrix-based formation scaling control},
   journal = {Automatica},
   volume = {101},
   pages = {120--127},
   DOI = {10.1016/j.automatica.2018.11.046},
   year = {2019}
}

@article{Zhou2022,
    author = {Zhou, Xin and Wen, Xiangyong and Wang, Zhepei and Gao, Yuman and Li, Haojia and Wang, Qianhao and Yang, Tiankai and Lu, Haojian and Cao, Yanjun and Xu, Chao and Gao, Fei},
    doi = {10.1126/scirobotics.abm5954},
    journal = {Science Robotics},
    number = {66},
    pages = {1--18},
    title = {Swarm of micro flying robots in the wild},
    volume = {7},
    year = {2022}
}

@article{DeMarina2021,
   author={de Marina, Hector Garcia},
   title = {Maneuvering and robustness issues in undirected displacement-consensus-based formation control},
   journal = {IEEE Transactions on Automatic Control},
   volume = {66},
   number = {7},
   pages = {3370-3377},
   DOI = {10.1109/tac.2020.3019780},
   year = {2021}
}

@article{Romero2024,
   author={Romero, Jose Guadalupe and Nuño, Emmanuel and Restrepo, Esteban and Sarras, Ioannis},
   title = {Global Consensus-Based Formation Control of Nonholonomic Mobile Robots With Time-Varying Delays and Without Velocity Measurements},
   journal = {IEEE Transactions on Automatic Control},
   volume = {69},
   number = {1},
   pages = {355-362},
   DOI = {10.1109/TAC.2023.3264744},
   year = {2024}
}

@article{Vu2024,
   author={Vu, Hieu Minh and Trinh, Minh Hoang and Van Tran, Quoc and Ahn, Hyo-Sung},
   title={Distance-based formation tracking of single- and double-integrator agents}, 
   journal = {IEEE Transactions on Automatic Control},
   volume = {69},
   number = {2},
   pages = {1332-1339},
   DOI = {10.1109/TAC.2023.3299817},
   year = {2024}
}

@article{HeXiaodong2024,
   author = {Xiaodong He and Zhongkui Li and Xiangke Wang and Zhiyong Geng},
   title = {Roto-translation invariant formation of fixed-wing {UAVs in 3D}: Feasibility and control},
   journal = {Automatica},
   volume = {161},
   pages = {111492},
   DOI = {10.1016/j.automatica.2023.111492},
   year = {2024}
}

@article{Asimow1978,
author = {Asimow, L and Roth, B},
journal = {Transactions of the American Mathematical Society},
title = {The rigidity of graphs},
Volume = {245},
Pages = {279-289},
DOI = {10.2307/1998867},
year = {1978}
}

@article{Quan2023,
   author = {Quan, L. and Yin, L. and Zhang, T. and Wang, M. and Wang, R. and Zhong, S. and Zhou, X. and Cao, Y. and Xu, C. and Gao, F.},
   title = {Robust and Efficient Trajectory Planning for Formation Flight in Dense Environments},
   journal = {IEEE Transactions on Robotics},
   volume = {39},
   number = {6},
   pages = {4785-4804},
   DOI = {10.1109/TRO.2023.3301295},
   year = {2023}
}

@article{Alonso2017,
   author = {Alonso-Mora, Javier and Baker, Stuart and Rus, Daniela},
   title = {Multi-robot formation control and object transport in dynamic environments via constrained optimization},
   journal = {The International Journal of Robotics Research},
   volume = {36},
   number = {9},
   pages = {1000-1021},
   DOI = {10.1177/0278364917719333},
   year = {2017}
}

@ARTICLE{KwangKyo2014,
  author={Oh, Kwang-Kyo and Ahn, Hyo-Sung},
  journal={IEEE Transactions on Automatic Control}, 
  title={Formation Control and Network Localization via Orientation Alignment}, 
  year={2014},
  volume={59},
  number={2},
  pages={540-545},
  doi={10.1109/TAC.2013.2272972}
}

@article{Song2017,
author = {Song, W. and Tang, Yutao and Hong, Y. and Hu, Xiaoming},
year = {2017},
pages = {4457-4477},
title = {Relative attitude formation control of multi-agent systems: RELATIVE ATTITUDE FORMATION CONTROL},
volume = {27},
number = {18},
journal = {International Journal of Robust and Nonlinear Control},
doi = {10.1002/rnc.3803}
}

@article{Dimarogonas2009,
title = {Leader–follower cooperative attitude control of multiple rigid bodies},
journal = {Systems \& Control Letters},
volume = {58},
number = {6},
pages = {429-435},
year = {2009},
doi = {10.1016/j.sysconle.2009.02.002},
author = {Dimos V. Dimarogonas and Panagiotis Tsiotras and Kostas J. Kyriakopoulos}
}

@INPROCEEDINGS{Wu2014,
  author={Wu, Tse-Huai and Lee, Taeyoung},
  booktitle={Proceedings of the 53th IEEE Conference on Decision and Control}, 
  title={Spacecraft position and attitude formation control using line-of-sight observations}, 
  address   = {Los Angeles, CA, USA},
  year={2014},
  pages={970-975},
  doi={10.1109/CDC.2014.7039507}
}

@ARTICLE{Arranz2014,
  author={Briñón-Arranz, Lara and Seuret, Alexandre and {Canudas-de-Wit}, Carlos},
  journal={IEEE Transactions on Automatic Control}, 
  title={Cooperative Control Design for Time-Varying Formations of Multi-Agent Systems}, 
  year={2014},
  volume={59},
  number={8},
  pages={2283-2288},
  doi={10.1109/TAC.2014.2303213}
}

@ARTICLE{Meng2025,
  author={Meng, Qingkai and Kasis, Andreas and Polycarpou, Marios M.},
  journal={IEEE Transactions on Aerospace and Electronic Systems}, 
  title={Integrated attitude-position formation control of multiple vehicles on {SE(3)} with individual objectives}, 
  year={2025},
  volume={61},
  number={3},
  pages={7710-7724},
  doi={10.1109/TAES.2025.3540976}
}

@ARTICLE{Zou2012,
  author={Zou, An-Min and Kumar, Krishna Dev},
  journal={IEEE Transactions on Aerospace and Electronic Systems}, 
  title={Distributed Attitude Coordination Control for Spacecraft Formation Flying}, 
  year={2012},
  volume={48},
  number={2},
  pages={1329-1346},
  doi={10.1109/TAES.2012.6178065}
}

@ARTICLE{Wei2018,
  author={Wei, Jieqiang and Zhang, Silun and Adaldo, Antonio and Thunberg, Johan and Hu, Xiaoming and Johansson, Karl H.},
  journal={IEEE Transactions on Automatic Control}, 
  title={Finite-Time Attitude Synchronization With Distributed Discontinuous Protocols}, 
  year={2018},
  volume={63},
  number={10},
  pages={3608-3615},
  doi={10.1109/TAC.2018.2797179}
}

@ARTICLE{Gao2023,
  author={Wu, Yuze and Yang, Fan and Wang, Ze and Wang, Kaiwei and Cao, Yanjun and Xu, Chao and Gao, Fei},
  journal={IEEE Robotics and Automation Letters}, 
  title={Ring-Rotor: A Novel Retractable Ring-Shaped Quadrotor With Aerial Grasping and Transportation Capability}, 
  year={2023},
  volume={8},
  number={4},
  pages={2126-2133},
  doi={10.1109/LRA.2023.3245499}
}

@ARTICLE{Zhou2025,
  author={Zhou, Xiaoren and Zhang, Meng and Hu, Jianchen and Wu, Chengshuai and Guan, Xiaohong},
  journal={IEEE Transactions on Industrial Electronics}, 
  title={A Fast {MEMS-IMU/GPS} In-Motion Alignment Method Using Full-Integration-Based Position Loci}, 
  year={2025},
  volume={72},
  number={11},
  pages={11812-11821},
  doi={10.1109/TIE.2025.3563700}
}

@ARTICLE{Vrba2025,
  author={Vrba, Matouš and Walter, Viktor and Pritzl, Václav and Pliska, Michal and Báča, Tomáš and Spurný, Vojtěch and Heřt, Daniel and Saska, Martin},
  journal={IEEE Transactions on Robotics}, 
  title={On Onboard LiDAR-Based Flying Object Detection}, 
  year={2025},
  volume={41},
  number={},
  pages={593-611},
  doi={10.1109/TRO.2024.3502494}
}

@article{Mercedes2021,
title = {Certifiable relative pose estimation},
journal = {Image and Vision Computing},
volume = {109},
pages = {104142},
year = {2021},
doi = {10.1016/j.imavis.2021.104142},
author = {Mercedes Garcia-Salguero and Jesus Briales and Javier Gonzalez-Jimenez}
}

@ARTICLE{Shin2024,
  author={Shin, Gihun and Sim, Hyunjae and Nam, Seungwon and Kim, Yonghee and Heo, Jae and Kim, Kwang-Ki K.},
  journal={IEEE Transactions on Intelligent Vehicles}, 
  title={Multi-robot relative pose estimation in {SE(2)} with observability analysis: A comparison of extended kalman filtering and robust pose graph optimization}, 
  year={2025},
  volume={10},
  number={8},
  pages={4293-4315},
  doi={10.1109/TIV.2024.3478290}
}

@article{Hector2021,
  title = {Distributed formation maneuver control by manipulating the complex Laplacian},
  journal = {Automatica},
  volume = {132},
  pages = {109813},
  year = {2021},
  doi = {10.1016/j.automatica.2021.109813},
  author={de Marina, Hector Garcia}
}

@ARTICLE{Zhao2024,
  author={Zhao, Yu and Gao, Kai and Huang, Panfeng and Chen, Guanrong},
  journal={IEEE Transactions on Automatic Control}, 
  title={Specified-Time Affine Formation Maneuver Control of Multiagent Systems Over Directed Networks}, 
  year={2024},
  volume={69},
  number={3},
  pages={1936-1943},
  doi={10.1109/TAC.2023.3309325}
}

@article{ZHANG2025,
title = {Linear formation control of multi-agent systems},
journal = {Automatica},
volume = {171},
pages = {111935},
year = {2025},
doi = {10.1016/j.automatica.2024.111935},
author = {Zhang, Xiaozhen and Yang, Qingkai and Xiao, Fan and Fang, Hao and Chen, Jie}
}

@ARTICLE{Fang2024,
  author={Fang, Xu and Xie, Lihua},
  journal={IEEE Transactions on Automatic Control}, 
  title={Distributed Formation Maneuver Control Using Complex Laplacian}, 
  year={2024},
  volume={69},
  number={3},
  pages={1850-1857},
  doi={10.1109/TAC.2023.3327932}
}

@ARTICLE{Huang2025,
  author={Huang, Yunchang and Dai, Shi-Lu},
  journal={IEEE Transactions on Control of Network Systems}, 
  title={Similarity-Based Rigidity Formation Maneuver Control of Underactuated Surface Vehicles Over Directed Graphs}, 
  year={2025},
  volume={12},
  number={1},
  pages={461-473},
  doi={10.1109/TCNS.2024.3487621}
}

@ARTICLE{Cheng2025,
  author={Cheng, Haoshu and Huang, Jie},
  journal={IEEE Transactions on Automatic Control}, 
  title={A General Framework for the Bearing-Based Formation Control}, 
  year={2025},
  volume={70},
  number={6},
  pages={3603-3616},
  doi={10.1109/TAC.2024.3513771}
}

@article{Anderson2009,
author = {Yu, Changbin and Anderson, Brian D. O. and Dasgupta, Soura and Fidan, Bari\c{s}},
title = {Control of Minimally Persistent Formations in the Plane},
journal = {SIAM Journal on Control and Optimization},
volume = {48},
number = {1},
pages = {206-233},
year = {2009},
doi = {10.1137/060678592}
}

@article{Yang2021,
author = {Yang, Junyi and Xiao, Feng and Chen, Tongwen},
title = {Formation Tracking of Nonholonomic Systems on the Special Euclidean Group under Fixed and Switching Topologies: An Affine Formation Strategy},
journal = {SIAM Journal on Control and Optimization},
volume = {59},
number = {4},
pages = {2850-2874},
year = {2021},
doi = {10.1137/20M1328130}
}

@ARTICLE{Li2018,
  author={Li, Xiuxian and Xie, Lihua},
  journal={IEEE Transactions on Automatic Control}, 
  title={Dynamic Formation Control Over Directed Networks Using Graphical Laplacian Approach}, 
  year={2018},
  volume={63},
  number={11},
  pages={3761-3774},
  doi={10.1109/TAC.2018.2798808}
}

@article{FANG2025,
author = {Fang, Xu and Xie, Lihua and Dimarogonas, Dimos V.},
journal = {Automatica},
title = {Simultaneous distributed localization and formation tracking control via matrix-weighted position constraints},
year = {2025},
volume = {175},
pages = {112188},
doi = {10.1016/j.automatica.2025.112188}
}

@ARTICLE{Chen20231,
  author={Chen, Liangming and Sun, Zhiyong},
  journal={IEEE Transactions on Automatic Control}, 
  title={Globally Stabilizing Triangularly Angle Rigid Formations}, 
  year={2023},
  volume={68},
  number={2},
  pages={1169-1175},
  doi={10.1109/TAC.2022.3151567}
}

@ARTICLE{Liu2024,
  author={Liu, Wenhang and Hu, Jiawei and Zhang, Heng and Wang, Michael Yu and Xiong, Zhenhua},
  journal={IEEE Transactions on Robotics}, 
  title={A Novel Graph-Based Motion Planner of Multi-Mobile Robot Systems With Formation and Obstacle Constraints}, 
  year={2024},
  volume={40},
  number={},
  pages={714-728},
  doi={10.1109/TRO.2023.3339989}
}

@ARTICLE{Culbertson2021,
  author={Culbertson, Preston and Slotine, Jean-Jacques and Schwager, Mac},
  journal={IEEE Transactions on Robotics}, 
  title={Decentralized Adaptive Control for Collaborative Manipulation of Rigid Bodies}, 
  year={2021},
  volume={37},
  number={6},
  pages={1906-1920},
  doi={10.1109/TRO.2021.3072021}
}

@ARTICLE{Chen2025,
  author={Chen, Liangming and Liang, Chenyang and Yuan, Shenghai and Cao, Muqing and Xie, Lihua},
  journal={IEEE Transactions on Robotics}, 
  title={Relative Localizability and Localization for Multirobot Systems}, 
  year={2025},
  volume={41},
  number={},
  pages={2931-2949},
  doi={10.1109/TRO.2025.3544103}
}

@ARTICLE{Erskine2024,
  author={Erskine, Julian and Briot, Sébastien and Fantoni, Isabelle and Chriette, Abdelhamid},
  journal={IEEE Transactions on Robotics}, 
  title={Singularity Analysis of Rigid Directed Bearing Graphs for Quadrotor Formations}, 
  year={2024},
  volume={40},
  number={},
  pages={139-157},
  doi={10.1109/TRO.2023.3324198}
}

@ARTICLE{Fathian2021,
  author={Fathian, Kaveh and Safaoui, Sleiman and Summers, Tyler H. and Gans, Nicholas R.},
  journal={IEEE Transactions on Robotics}, 
  title={Robust Distributed Planar Formation Control for Higher Order Holonomic and Nonholonomic Agents}, 
  year={2021},
  volume={37},
  number={1},
  pages={185-205},
  doi={10.1109/TRO.2020.3014022}
}

@article{HeGen2025,
title = {Global stabilization of similar formation via edge-based clique addition},
journal = {Automatica},
volume = {173},
pages = {112084},
year = {2025},
author = {Gen He and Gangshan Jing and Yongduan Song}
}

@ARTICLE{ZhangFan2025,
  author={Zhang, Songyuan and So, Oswin and Garg, Kunal and Fan, Chuchu},
  journal={IEEE Transactions on Robotics}, 
  title={{GCBF+}: A Neural Graph Control Barrier Function Framework for Distributed Safe Multiagent Control}, 
  year={2025},
  volume={41},
  number={},
  pages={1533-1552},
  doi={10.1109/TRO.2025.3530348}
}

@ARTICLE{Kratky2025,
  author={Kratky, Vit and Penicka, Robert and Horyna, Jiri and Stibinger, Petr and Baca, Tomas and Petrlik, Matej and Stepan, Petr and Saska, Martin},
  journal={IEEE Transactions on Robotics}, 
  title={{CAT-ORA}: Collision-Aware Time-Optimal Formation Reshaping for Efficient Robot Coordination in {3-D} Environments}, 
  year={2025},
  volume={41},
  number={},
  pages={2950-2969},
  doi={10.1109/TRO.2025.3547296}
}

@ARTICLE{Silva2020,
  author={Silva, Bruno and Hancke, Gerhard P.},
  journal={IEEE Transactions on Industrial Informatics}, 
  title={Ranging Error Mitigation for Through-the-Wall Non-Line-of-Sight Conditions}, 
  year={2020},
  volume={16},
  number={11},
  pages={6903-6911},
  doi={10.1109/TII.2020.2969886}
}

@ARTICLE{Zhu2023,
  author={Zhu, Zhihao and Guo, Shisheng and Chen, Jiahui and Xue, Shucheng and Xu, Zihan and Wu, Peilun and Cui, Guolong and Kong, Lingjiang},
  journal={IEEE Transactions on Instrumentation and Measurement}, 
  title={Non-Line-of-Sight Targets Localization Algorithm via Joint Estimation of {DoD} and {DoA}}, 
  year={2023},
  volume={72},
  number={},
  pages={1-11},
  doi={10.1109/TIM.2023.3323003}
}

@article{Hossein2022,
title = {A spatio-temporal reference trajectory planner approach to collision-free continuum deformation coordination},
journal = {Automatica},
volume = {142},
pages = {110255},
year = {2022},
doi = {10.1016/j.automatica.2022.110255},
author = {Hossein Rastgoftar and Ilya V. Kolmanovsky}
}

@ARTICLE{Aranda2025,
  author={Aranda, Miguel and Cuiral-Zueco, Ignacio and López-Nicolás, Gonzalo},
  journal={IEEE Control Systems Letters}, 
  title={Distributed Control of Flexible Chained Multiagent Formations}, 
  year={2025},
  volume={9},
  number={},
  pages={2018-2023},
  doi={10.1109/LCSYS.2025.3590428}
}

\begin{IEEEbiography}[{\includegraphics[width=1.2in,height=1.35in,clip,keepaspectratio]{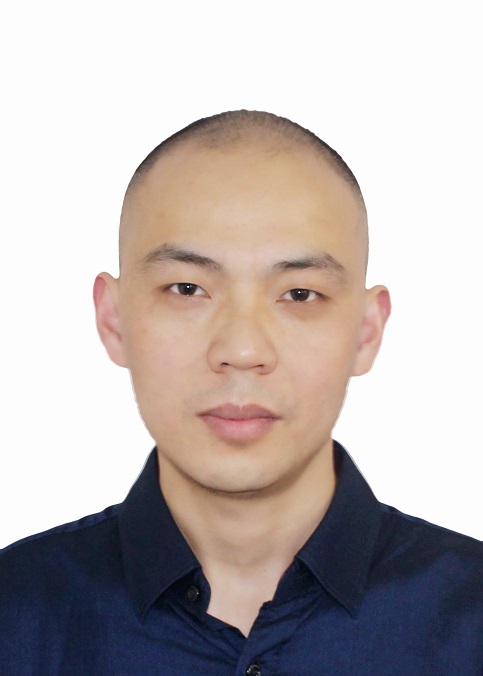}}]{Tao He} received the B.S. degree in ‌Electronic and Information Engineering from Chongqing University, Chongqing, China, in 2009 and his M.S. degree in Computer Science from the University of Electronic Science and Technology of China, Chengdu, China, in 2023. He is currently pursuing the Ph.D. degree in the School of Automation, Chongqing University, Chongqing, China. His research interests include cooperative control and motion planning for multi-agent systems.
\end{IEEEbiography} 

\begin{IEEEbiography}[{\includegraphics[width=1in,height=1.25in,clip,keepaspectratio]{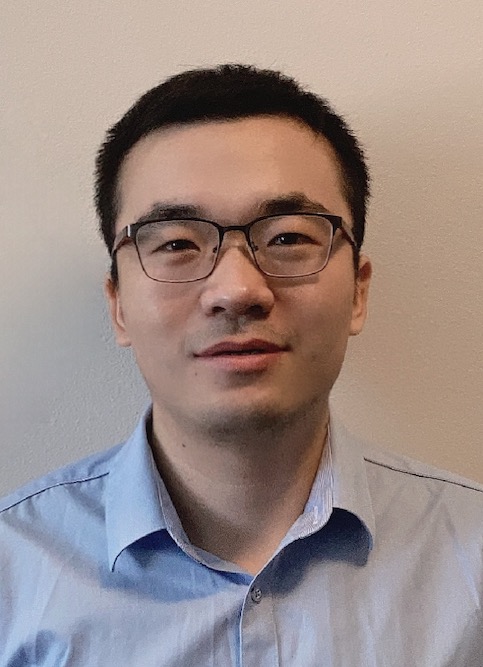}}]{Gangshan Jing} received the Ph.D. degree in Control Theory and Control Engineering from Xidian University, Xi'an, China, in 2018. From 2016-2017, he was a research assistant at Hong Kong Polytechnic University. From 2018 to 2019, he was a postdoctoral researcher at Ohio State University. From 2019 to 2021, he was a postdoctoral researcher at North Carolina State University. Since 2021 Dec., he has been a professor with the School of Automation, Chongqing University. His research interests include cooperative control, optimization, and learning for network systems.
\end{IEEEbiography}

\end{document}